\documentclass[a4paper]{article}

\usepackage[english]{babel}
\usepackage[utf8x]{inputenc}
\usepackage[T1]{fontenc}
\usepackage{graphicx}
\usepackage{caption, subcaption}
\usepackage[title]{appendix}
\usepackage{tikz}
\usepackage{pgfplots}
\usepackage{cutwin}
\graphicspath{ {images/} }
\usetikzlibrary{intersections,decorations.pathreplacing,decorations.markings,calc,angles,quotes,arrows.meta,pgfplots.fillbetween,patterns}

\usepackage[a4paper,top=2.8cm,bottom=3cm,left=2.7cm,right=2.7cm,marginparwidth=1.75cm]{geometry}

\usepackage{amsmath,yhmath,amsfonts}
\usepackage{amsthm}
\usepackage{amssymb}
\usepackage[hidelinks]{hyperref}
\usepackage{cleveref}
\usepackage{epstopdf}
\usepackage{float}
\usepackage[sort&compress,numbers]{natbib}
\newtheorem{thm}{Theorem}
\newtheorem{lem}[thm]{Lemma}
\newtheorem{prop}[thm]{Proposition}

\theoremstyle{definition}
\newtheorem{defn}[thm]{Definition}
\newtheorem{rmk}[thm]{Remark}
\numberwithin{equation}{section}
\numberwithin{thm}{section}

\newcommand{\B}{\mathcal{B}}
\newcommand{\C}{\mathcal{C}}
\newcommand{\Cb}{\mathbb{C}}
\newcommand{\D}{\mathcal{D}}
\newcommand{\F}{\mathcal{F}}

\newcommand{\N}{\mathbb{N}}

\newcommand{\R}{\mathbb{R}}

\renewcommand{\S}{\mathcal{S}}
\newcommand{\T}{\mathcal{T}}
\newcommand{\I}{\mathcal{I}}
\newcommand{\U}{\mathcal{U}}

\newcommand{\Z}{\mathbb{Z}}
\newcommand{\p}{\partial}
\renewcommand{\epsilon}{\varepsilon}
\newcommand{\dx}{\: \mathrm{d}}

\renewcommand{\u}{\mathbf{u}}
\newcommand{\uf}{\mathfrak{u}}
\newcommand{\Bf}{\mathfrak{B}}
\newcommand{\Cf}{\mathfrak{C}}

\newcommand{\Hf}{\mathfrak{H}}

\newcommand{\ie}{\textit{i.e.}}
\newcommand{\nm}{\noalign{\smallskip}}
\newcommand{\ds}{\displaystyle}
\newcommand{\iu}{\mathrm{i}\mkern1mu}
\newcommand{\Dc}{\mathcal{D}}

\newcommand{\svdots}{\raisebox{3pt}{\scalebox{.6}{$\vdots$}}}
\newcommand{\sddots}{\raisebox{3pt}{\scalebox{.6}{$\ddots$}}}
\newcommand{\sadots}{\raisebox{3pt}{\scalebox{.6}{$\adots$}}}

\makeatletter
\newcommand{\neutralize}[1]{\expandafter\let\csname c@#1\endcsname\count@}
\makeatother

\title{Anderson localization in the subwavelength regime}
\author{
	Habib Ammari\thanks{\footnotesize Department of Mathematics,
		ETH Z\"urich, Z\"urich, Switzerland (habib.ammari@math.ethz.ch).}\and Bryn Davies\thanks{\footnotesize Department of Mathematics, Imperial College London, London, UK (bryn.davies@imperial.ac.uk).} \and Erik Orvehed Hiltunen\thanks{\footnotesize Department of Mathematics, Yale University, New Haven, USA (erik.hiltunen@yale.edu).}}
\date{}

\begin{document}
	\maketitle

	\begin{abstract}
	In this paper, we use recent breakthroughs in the study of coupled subwavelength resonator systems to reveal new insight into the mechanisms responsible for the fundamental features of Anderson localization. The occurrence strong localization in random media has proved difficult to understand, particularly in physically derived multi-dimensional models and systems with long-range interactions. We show here that the scattering of time-harmonic waves by high-contrast resonators with randomly chosen material parameters reproduces the characteristic features of Anderson localization. In particular, we show that the hybridization of subwavelength resonant modes is responsible for both the repulsion of energy levels as well as the widely observed phase transition, at which point eigenmode symmetries swap and very strong localization is possible. We derive results from first principles, using asymptotic expansions in terms of the material contrast parameter and obtain a characterization of the localized modes in terms of generalized capacitance matrices. This model captures the long-range interactions of the wave-scattering system and provides a concise framework to explain the exotic phenomena that are observed.
	\end{abstract}
\vspace{0.5cm}
	\noindent{\textbf{Mathematics Subject Classification (MSC2010):} 35J05, 35C20, 35P20, 78A48.

\vspace{0.2cm}

\vspace{0.5cm}

	\section{Introduction} 
	
	As waves propagate through our environment they often encounter materials with random structures. These random media might be organic substances, fluids, or structured materials that have been subjected to random perturbations. As a result, there is a large field concerned with the propagation of waves in random media, with lots of attention being devoted to the eye-catching phenomenon of wave localization \cite{ishimaru1978wave, fouque2007wave}. The seminal work in this area was the observation by Philip Anderson in 1958 that random media could be free from diffusion \cite{anderson1958absence}. His work concerned the transport of electrons in random lattices, but has inspired studies in many different wave regimes \cite{lagendijk2009fifty, segev2013anderson, billy2008direct, filoche2012universal, crane2017anderson, carmona1982exponential, davies2023landscape, torres2019level}.

	Wave localization is typically constrained by the diffraction limit. Anderson localization is generally not an exception to this and the localization length is typically limited by the wavelength, meaning it diverges in the low-frequency regime. As a result, in this classical setting, inhomogeneities on length-scales smaller than the wavelength will have a negligible effect on wave scattering and wave transport is characterized by effective parameters where any small disorder can be neglected. However, recent developments of systems exhibiting \emph{subwavelength} wave localization are able to overcome this limit \cite{ammari2021functional, sheinfux2017observation, herzig2016interplay, minnaert1933musical}. Crucially, this is achieved in \emph{locally resonant} materials, whose constituents exhibit resonance at small (\textit{i.e.} \emph{subwavelength}) frequencies. In particular, \emph{high-contrast} resonators are a natural example of subwavelength resonance (see \textit{e.g.} \cite{ammari2021functional} for a comprehensive review of related phenomena; a famous example is the Minnaert resonance of air bubbles in water \cite{minnaert1933musical}). 
	
	While random systems of subwavelength resonators are challenging to understand, the spectra of periodic systems are much more straightforward. A periodic crystal of $N$ repeated high-contrast resonators is known to have $N$ subwavelength resonant band functions. A band gap is guaranteed to exist above these bands and other gaps can exist between them, depending on the choice of geometry \cite{ammari2021functional}. Perturbations can create eigenmodes with eigenfrequencies within these band gaps. These `midgap' modes are elements of the point spectrum and cannot couple with the Bloch modes so decay exponentially quickly away from the defect. Some examples of localized modes, for a planar square lattice of high-contrast three-dimensional resonators with randomly chosen material parameters are shown in \Cref{fig:modesA}. We can derive concise formulas for the the midgap frequencies that exist when there are finitely many random perturbations and can describe the localized modes of a fully random structure by increasing the number of perturbations.
	
	The behaviour of the randomly perturbed systems studied here can be understood by considering the simple phenomena of \emph{hybridization} and \emph{level repulsion}. That is, when two eigenmodes are coupled, one of the eigenfrequencies of the resulting hybridized modes will be shifted up while the other will be shifted down. Meanwhile, level repulsion means that when disorder is added to the entries of a matrix, the eigenvalues have a tendency to separate \cite{mehta2004random}. These can be used to explain, firstly, the observation that the stronger the disorder, the more localized modes are created, as well as the tendency for the eigenmodes to become more strongly localized as the strength of the disorder increases.	The most commonly studied random matrices are Gaussian ensembles, for which the distributions of eigenvalue separations are well known \cite{mehta2004random}. In our setting, however, we have a non-linear eigenvalue problem so are not able to use this theory. As a result, we will elect to study the simple case of uniformly distributed perturbations (however, we consistently noticed qualitatively very similar results to the case of Gaussian perturbations).
			
	The phase transition whereby disordered media changes from being conducting (exhibiting only weak localization) to being insulating (the absence of diffusion) is referred to as the \emph{Anderson transition} \cite{anderson1958absence, lagendijk2009fifty}. In this work, we similarly observe a phase transition at a positive value of the perturbation strength. At this transition, eigenfrequencies become degenerate and the symmetries of the corresponding modes swap (\emph{i.e.}, even modes change to odd modes, and vice versa). In general, when eigenmodes hybridize, the volume of their support increases, decreasing the extent to which they can be said to be localized. This trend, however, breaks down at the phase transition. At this point, any small perturbation will lift the degeneracy and cause the modes to decouple, producing strongly localized modes. A sharp peak in the average degree of localization is observed here, which is a key signature of the transition.
	
	The analysis in this work is based on the \emph{generalized capacitance matrix}, which is a powerful tool for characterising the subwavelength resonant modes of a system of high-contrast resonators \cite{ammari2021functional}. In essence, the capacitance formulation provides a discrete approximation to the continuous spectral problem of the PDE-model, valid in the high-contrast asymptotic limit. This approximation is based solely on first principles and provides a natural starting point for both theoretical analysis and numerical simulation of Anderson localization. The inclusion of \emph{long-range interactions} is a crucial feature of the capacitance matrix: the off-diagonal entries $C_{ij}$ decay slowly, according to $1/|i-j|$. In contrast, previous studies of Anderson localization in quantum-mechanical and Hamiltonian systems are mainly limited to the case of short-range interactions $1/|i-j|^{d+\epsilon}$, where $d$ is the dimension and $\epsilon > 0$ \cite{anderson1958absence, figotin1994localization, aizenman1993localization}. Since Anderson's original work \cite{anderson1958absence}, localization was believed to only occur in systems with short-range interactions but, more recently, localization has been demonstrated in systems with long-range interactions \cite{nandkishore2017many, jenkins2018strong}. Nevertheless, the long-range interactions and the dense random perturbations in the setting studied here make the problem considerably more challenging than Anderson's original model, and precise statements on the spectrum of the discrete model are generally beyond reach.
	
 	Anderson localization has been studied in other systems related to the current setting, most notably for classical waves in \cite{figotin1994localization}. However, the analysis is typically based on approximating the continuous (differential) equation by a discrete (difference) equation. In contrast, the novelty of our approach is that the discrete capacitance approximation studied in this work is derived from first principles and is a natural approximation based on the physical properties of the system. It is important to emphasize that tight-binding and nearest-neighbour formulations are not applicable in the setting of high-contrast resonators and give poor approximations \cite{frohlich1983absence, frohlich1985constructive, ammari2021validity}.
 	
	The analysis in this work deepens the understanding of localization in systems with long-range interactions. We begin by describing the Helmholtz partial differential model of interest in \Cref{sec:setup} and derive the corresponding generalized capacitance formulation in \Cref{sec:cap}. Within this framework, the Helmholtz problem is approximated by a spectral problem for a discrete operator, where the random disorder is manifested through multiplication by a diagonal operator. The capacitance formulation leads to a concise characterization of localized modes in the case of finitely many defects in an infinite structure, whereby the problem reduces to a nonlinear eigenvalue problem for a Toeplitz matrix. This is presented in \Cref{sec:defect} and allows for a detailed study of level repulsion and phase transition, providing insight into the mechanism of Anderson localization in systems with long-range interactions. Additionally, we study the limit when the number of defects increases, demonstrating an increasing number of localized modes. Finally, in \Cref{sec:fullyrandom} we consider the fully random case, \textit{i.e.} when all resonators in a finite array are taken to have random parameters. As observed in discrete systems with short-range interactions (see, \textit{e.g.}, \cite{figotin1994localization}), the band structure of the unperturbed case persists, while localized modes emerge around the edges of the bands. As the perturbation increases, the level repulsion will cause the localized frequencies to separate and increase the degree of localization.

	\begin{figure}
		\begin{center}
			\includegraphics[width=0.7\linewidth]{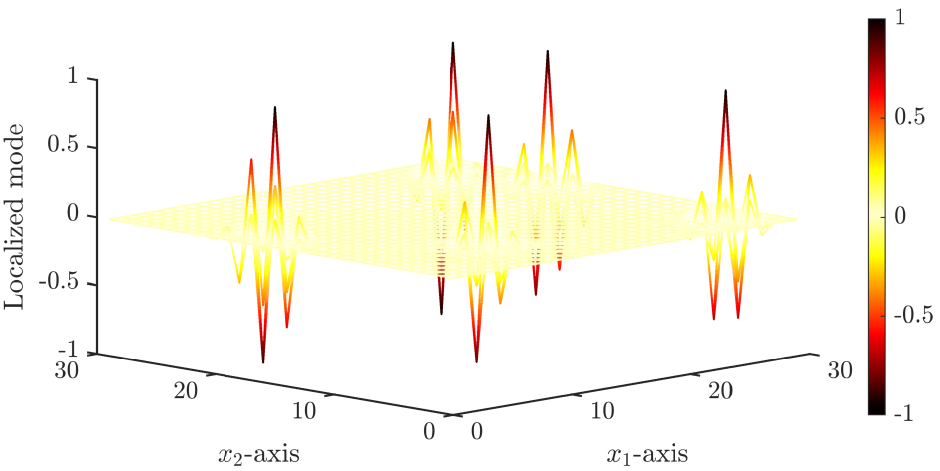}
		\end{center}
		\caption{Examples of localized states in a two-dimensional, finite, square lattice of resonators. We show five different localized eigenmodes (superimposed on one another) corresponding to a system of $N=900$ resonators whose material parameters are perturbed randomly according to a uniform distribution with standard deviation $\sigma = 0.05$; see \Cref{sec:fullyrandom} for details of the setup.} \label{fig:modesA}
	\end{figure}

	\section{Formulation of the resonance problem} \label{sec:setup}
	
	\begin{figure}
		\begin{subfigure}{\linewidth}
		\centering
				\includegraphics[width=0.8\linewidth]{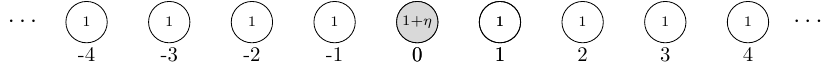}
				\caption{} \label{fig:arraysketchone}
		\end{subfigure}
		
		\vspace{0.4cm}
		
		\begin{subfigure}{\linewidth}
		\centering
				\includegraphics[width=0.92\linewidth]{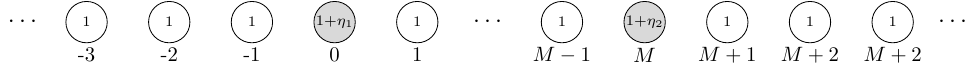}
				\caption{} \label{fig:arraysketchtwo}
		\end{subfigure}	
		
		\vspace{0.4cm}
		
		\begin{subfigure}{\linewidth}
		\centering
				\includegraphics[width=0.88\linewidth]{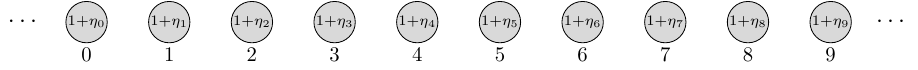}
				\caption{} \label{fig:arraysketchfull}
		\end{subfigure}	
			
		\vspace{0.4cm}
		
		\begin{subfigure}{\linewidth}
		\centering
				\includegraphics[width=0.88\linewidth]{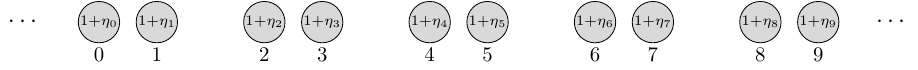}
				\caption{} \label{fig:arraysketchdimers}
		\end{subfigure}	
		\caption{Arrays of resonators with (a) a single defect, on the 0\textsuperscript{th} resonator, (b) two defects, on the 0\textsuperscript{th} and the $M$\textsuperscript{th} resonators, (c) defects on every resonator, and (d) defects on every resonator in an array of dimers. In each case, the material parameters on defective resonators are perturbed $1\mapsto1+\eta_i$ for some $\eta_i\in(-1,\infty)$. Here, $\eta_i$ are independent of the high-contrast parameter $\delta$. We depict circular resonators here, but the theoretical results in this work hold for any shape with H\"older continuous boundary. The main idea behind this work is to conduct analytic studies of structures (a) and (b) to make predictions about the behaviour of the fully random systems (c) and (d).}
		\label{fig:arraysketch}
	\end{figure}

	We assume that the resonator system consists of periodically arranged scatterers, whose material parameters may be non-periodic. We work in $d$ spatial dimensions, where $d$ can be either 2 or 3 (Anderson localization in one-dimensional systems is a very different and somewhat more straightforward phenomenon). We consider a lattice $\Lambda$ of dimension $1\leq d_l\leq d$, generated by the lattice vectors $l_1,..., l_{d_l}$:
	$$\Lambda := \left\{ m_1 l_1+...+m_{d_l} l_{d_l} ~|~ m_i \in \Z \right\}. $$
	
	We let $P_l: \R^d \to \R^{d_l}$ be the projection onto the first $d_l$ coordinates, and $P_\perp: \R^d \to \R^{d-d_l}$ be the projection onto the last $d-d_l$ coordinates. For simplicity, we assume that $P_\perp l_i = 0$, which means that the lattice is aligned with the first $d_l$ coordinate axes. For a point $x\in \R^d$, we use the notation $x = (x_l,x_0)$, where $x_l =P_lx$ and $x_0 = P_\perp x$.
	We let $Y \subset \R^{d}$ be a fundamental domain of the given lattice:
	$$ Y := \left\{ c_1 l_1+...+c_{d_l} l_{d_l} ~|~ 0 \le c_1,...,c_{d_l} \le 1 \right\}. $$
	Inside $Y$, we assume that there are $N$ scatterers, or \emph{resonators}, $D_i\subset Y$, for $i=1,...,N$. These are disjoint, connected domains with boundaries in $C^{1,s}$ for some $0<s<1$.	For $m\in \Lambda$, we let $D_i^m$ denote the translated resonator, and we let $\D$ denote the total collection of resonators:
	$$\D = \bigcup_{m\in \Lambda}\bigcup_{i\in\{1,\dots, N\}} D_i^m, \qquad D_i^m = D_i+m.$$
	We let $v_i^m$ denote the wave speed in $D_i^m$, and $v$ denote the wave speed in the surrounding medium. Each resonator is additionally associated to a parameter $\delta_i^m$, which describes the contrast between the resonator and the surrounding.
		
	Throughout this work, we will consider the equation
	\begin{equation} \label{eq:scattering}
		\left\{
		\begin{array} {ll}
			\ds \Delta {u}+ \frac{\omega^2}{v^2} {u}  = 0 & \text{in } \R^d \setminus \Dc, \\[0.3em]
			\ds \Delta {u}+ \frac{\omega^2}{(v_i^m)^2} {u}  = 0 & \text{in } D_i^m, \ i=1,\dots,N, \\
			\nm
			\ds  {u}|_{+} -{u}|_{-}  = 0  & \text{on } \partial \Dc, \\
			\nm
			\ds  \delta_i^m \frac{\partial {u}}{\partial \nu} \bigg|_{+} - \frac{\partial {u}}{\partial \nu} \bigg|_{-} = 0 & \text{on } \partial D_i^m, \ i=1,\dots,N, \\
			\nm
			\ds u(x_l,x_0) & \text{satisfies the outgoing radiation condition as }  |x_0| \rightarrow \infty.
		\end{array}
		\right.
	\end{equation}
	
	The natural setting to achieve subwavelength resonance is to consider the \emph{high-contrast} case  \cite{ammari2021functional}, whereby we assume that 
	$$\delta_i^m = O(\delta)$$
	for some small parameter $\delta \to 0$. Within this high-contrast regime, we are interested in \emph{subwavelength} solutions to \eqref{eq:scattering}. These are defined as non-trivial solutions for which the frequency $\omega=\omega(\delta)$ depends continuously on $\delta$ and has the asymptotic behaviour
	\begin{equation} \label{subwave}
	\omega\to0\quad\text{as}\quad\delta\to0.
	\end{equation}
	This asymptotic condition is the notion of the ``subwavelength'' regime that we consider in this work.

 	The dual lattice of $\Lambda$, denoted $\Lambda^*$, is generated by $\alpha_1,...,\alpha_{d_l}$ satisfying $ \alpha_i\cdot l_j = 2\pi \delta_{ij}$ and $P_\perp \alpha_i = 0$,  for $i,j = 1,...,d_l.$ The \emph{Brillouin zone} $Y^*$ is defined as $Y^*:= \big(\R^{d_l}\times\{\mathbf{0}\}\big) / \Lambda^*$, where $\mathbf{0}$ is the zero-vector in $\R^{d-d_l}$. We remark that $Y^*$ can be written as $Y^*=Y^*_l\times\{\mathbf{0}\}$, where  $Y^*_l$ has the topology of a torus in $d_l$ dimensions.

	\section{Generalized capacitance matrices}\label{sec:cap}
	In this section, we introduce the capacitance matrix formalism, which characterizes the subwavelength spectrum of \eqref{eq:scattering} in terms of a discrete operator. In its simplest form, the generalized quasiperiodic capacitance matrix $\widehat{\C}^\alpha$ gives a concise characterization of the periodic problem. We then generalize this framework to a general, possibly random, parameter distribution. Similar characterizations, in terms of the generalized capacitance matrix, have been used to study the subwavelength resonant modes of many different systems of high-contrast resonators, see \cite{ammari2021functional} for a review. In the case of a compact perturbation of a periodic system, this capacitance formalism reduces to an eigenvalue problem for a Toeplitz matrix.

	\subsection{Layer potentials and the quasiperiodic capacitance matrix}
	In the case of a periodic system, the spectrum of \eqref{eq:scattering} is well-known to consist of a countable sequence of continuous spectral bands. These bands consist of propagating ``Bloch'' modes that are not spatially localized. The general theory for periodic elliptic operators is reviewed in \cite{kuchment2016overview} and presented for the specific case of periodic subwavelength resonators in \cite{ammari2021functional}. In this section, we briefly summarize the main results. We assume that the material parameters satisfy
	$$\delta_i^m = \delta_i, \qquad v_i^m = v_i, \qquad \text{for all } m \in \Lambda,$$
	for some constants $\delta_i, v_i$, $i=1,...,N$.
	
	The capacitance matrix describes the band functions in the subwavelength limit, whereby the Helmholtz equation reduces to Laplace's equation. To define this matrix, we begin by letting $G(x)$ be the Green's function for Laplace's equation in $d=2$ or $d=3$:
	\begin{equation} \label{eq:G}
		G(x)=
		\left\{
		\begin{array}{l l}
			\ds	-\frac{1}{2\pi} \log(|x|), & d=2,\\
			\ds -\frac{1}{4\pi|x|}, & d=3,
		\end{array}\right.
		\qquad x\neq0.
	\end{equation}
	For $\alpha \in  Y\setminus \{0\}$, we can then define the quasiperiodic Green's function $G^{\alpha}(x)$ as 
	\begin{equation}\label{eq:xrep}
		G^{\alpha}(x) := \sum_{m \in \Lambda} G(x-m)e^{\iu \alpha \cdot m}.
	\end{equation}
	The series in \eqref{eq:xrep} converges uniformly for $x$ and $y$ in compact sets of $\R^d$, $x\neq y$ and $\alpha \neq 0$.  The quasiperiodic single layer potential  $\mathcal{S}_D^{\alpha}$ is defined as
	$$\mathcal{S}_D^{\alpha}[\varphi](x) := \int_{\partial D} G^{\alpha} (x-y) \varphi(y) \dx\sigma(y),\quad x\in \mathbb{R}^d.$$
	Restricting to $x\in \p D$, $\S_D^\alpha$ is an invertible operator $\S_D^\alpha: L^2(\p D) \to H^1(\p D)$ for $\alpha \neq 0$ (see, for example, \cite{ammari2021functional} for a review of layer potential operators with applications to subwavelength physics}).
	\begin{defn}[Quasiperiodic capacitance matrix] \label{defn:QCM}
		Assume $\alpha \neq 0$. For a system of $N\in\N$ resonators $D_1,\dots,D_N$ in $Y$ we define the quasiperiodic capacitance matrix $\widehat{C}^\alpha=(\widehat{C}^\alpha_{ij})\in\mathbb{C}^{N\times N}$ to be the square matrix given by
		\begin{equation*}
			\widehat{C}^\alpha_{ij}=-\int_{\p D_i} (\S_D^{\alpha})^{-1}[\chi_{\p D_j}] \dx \sigma,\quad i,j=1,\dots,N.
		\end{equation*}
	\end{defn}
	At $\alpha = 0$, we can define $\widehat{C}^0$ to be the limit of $\widehat{C}^\alpha$ as $\alpha \to 0$, which is known to exist \cite{ammari2021functional}. Having defined the \emph{capacitance} matrix, we now introduce the \emph{generalized} capacitance matrix, which allows us to account for the material parameters of the resonators.
	
	\begin{defn}[Generalized quasiperiodic capacitance matrix] \label{defn:GQCM}
		For a system of $N\in\N$ resonators $D_1,\dots,D_N$ in $Y$ we define the generalized quasiperiodic capacitance matrix, denoted by $\widehat{\C}^\alpha=(\widehat{\C}^\alpha_{ij})\in\mathbb{C}^{N\times N}$, to be the square matrix given by
		\begin{equation*}
			\widehat{\C}^\alpha_{ij}=\frac{\delta_i v_i^2}{|D_i|} \widehat{C}^\alpha_{ij}, \quad i,j=1,\dots,N.
		\end{equation*}
	\end{defn}

	Since \eqref{eq:scattering} is periodic, it is well-known that the spectrum consists of continuous bands parametrized by $\alpha$. For small $\delta$, there are $N$ such bands $\omega_i(\alpha)$, $i=1,...,N$ satisfying $\omega_i(\alpha) = O(\delta^{1/2})$. We have the following result \cite{ammari2021functional}.
	\begin{thm} \label{thm:res_periodic}
		Let $d\in \{2,3\}$ and $0<d_l\leq d$. Consider a system of $N$ subwavelength resonators in the periodic unit cell $Y$, and assume $|\alpha| > c > 0$ for some constant $c$ independent of $\omega$ and $\delta$. As $\delta\to0$, there are $N$ subwavelength resonant frequencies, which satisfy the asymptotic formula
		\begin{equation*}
			\omega_n^\alpha = \sqrt{\lambda_n^\alpha}+O(\delta^{3/2}), \quad n=1,\dots,N.
		\end{equation*}
		Here, $\{\lambda_n^\alpha: n=1,\dots,N\}$ are the eigenvalues of the generalized quasiperiodic capacitance matrix $\widehat{\C}^\alpha\in\mathbb{C}^{N\times N}$, which satisfy $\lambda_n^\alpha=O(\delta)$ as $\delta\to0$.
	\end{thm} 
	The spectrum of the unperturbed periodic array consists of a countable collection of spectral bands $\omega^\alpha$, which depend continuously on $\alpha\in Y^*$. The bands which are subwavelength, in the sense of the asymptotic condition \eqref{subwave}, are described by the eigenvalues of the generalized quasiperiodic capacitance matrix $\widehat{\C}^\alpha$ thanks to \Cref{thm:res_periodic}. These bands consist of propagating ``Bloch'' modes and might have band gaps between them.  The introduction of (random) perturbations to the model will, in some cases, create localized modes with resonant frequencies lying in one of the band gaps. One advantage of the subwavelength regime we study is that there will always be a band gap above the top subwavelength spectral band, provided $\delta$ is sufficiently small. See \cite{ammari2017BGopening} for a detailed exposition of this phenomenon. Many of the examples considered in this work are of this type, whereby localized modes are perturbed away from the top of the subwavelength resonant spectrum due to the introduction of defects. We will also consider structures that have band gaps between the subwavelength bands, such as the dimer array depicted in \Cref{fig:arraysketchdimers}.
	
	In this work, we consistently use the superscript $\ \widehat{} \ $ to denote quantities posed in terms of the dual variable $\alpha\in Y^*$. In subsequent sections, we will re-frame the capacitance formulation to the real-space variable $m \in \Lambda$. We introduce the Floquet transform $\F: \ \bigr(\ell^2(\Lambda)\bigr)^N \to \left(L^2(Y^*)\right)^N$ and its inverse $\I: \ \left(L^2(Y^*)\right)^N \to \bigr(\ell^2(\Lambda)\bigr)^N $, which are given by
	$$\F[\phi](\alpha) := \sum_{m\in \Lambda} \phi(m)e^{\iu\alpha\cdot m}, \qquad  \I[\psi](m) := \frac{1}{|Y^*|}\int_{Y^*} \psi(\alpha) e^{-\iu \alpha \cdot m} \dx \alpha.$$
	We can now define the real-space capacitance coefficients and generalized capacitance coefficients, respectively, as
	$$C^m = \I[\widehat{C}^\alpha](m), \qquad \C^m = \I[\widehat{\C}^\alpha](m),$$	
	indexed by the real-space variable $m\in \Lambda$.

	\subsection{Discrete-operator formulation}
	The goal, now, is to derive a capacitance formulation which is valid even in the case when the material parameters are not distributed periodically. 
	
	We seek solutions $u$ to \eqref{eq:scattering} which are localized and in the subwavelength regime, as defined by the asymptotic condition \eqref{subwave}. It has been shown previously that for a resonant mode $u$ to exist within this regime, it must have an eigenfrequency $\omega$ with the asymptotic behaviour $\omega=O(\sqrt{\delta})$ \cite{ammari2021functional}. Hence, we will write $\omega$ as
		\begin{equation}\label{eq:omega}
		\omega = \omega_0 + O(\delta), \quad \omega_0 = \beta \delta^{1/2},
	\end{equation}
	for some constant $\beta$ independent of $\delta$. Here, we refer to localization in the $L^2$-sense, meaning that modes are said to be localized if they decay sufficiently quickly such that
	$$
	\int_{\R^{d_l}} |u(x_l,x_0)|^2 \dx x_l<\infty,
	$$
	for any $x_0\in \R^{d-d_l}$. For simplicity, we also assume that $u$ is normalized in the sense
	$$\int_{\R^{d_l}} |u(x_l,0)|^2 \dx x_l=1.$$ 
	The crucial property of subwavelength solutions is that they are approximately constant as functions of $x$ inside the resonators $D_i^m$; such solutions arise as perturbations of the constant solution of the interior Neumann eigenvalue problem of $-\Delta$ inside $D_i^m$. Specifically, we have the following result from \cite{pt_topological}.
	\begin{lem}\label{lem:constant}
		Let $u$ be a localized, normalized eigenmode to \eqref{eq:scattering} corresponding to an eigenvalue $\omega$ which satisfies $\omega = O(\delta^{1/2})$ as $\delta \to 0$. Then, uniformly in $x\in \D$,
		$$u(x) = u_i^m + O(\delta^{1/2}), \quad x\in D_i^m, i=1,...,N, m\in \Lambda, $$
		for some $u_i^m$ which are constant with respect to position $x$ (within each disjoint $D_i^m$) and material contrast $\delta$.
	\end{lem}
	We define $\u^m = \left(\begin{smallmatrix} u_1^m \\ \svdots \\ u_N^m\end{smallmatrix}\right)\in \R^N$.	We denote the material parameters as
	$$\delta_i^m(v_i^m)^2 = \delta_iv_i^2b_i^m$$
	for some (possibly random) variables $b_i^m$ which satisfy $b_i^m = O(1)$ as $\delta \to 0$. We define $\B_m$ as the $N\times N$ diagonal matrix whose $i$\textsuperscript{th} entry is given by $b_i^m$.

The Floquet transform is the tool that allows us to show the following result, by transforming into the dual space and then back again.
	\begin{prop} \label{prop:loc_main}
		Any localized solution $u$ to \eqref{eq:scattering}, corresponding to a subwavelength frequency $\omega = \omega_0 + O(\delta)$, satisfies 		\begin{equation} \label{eq:anderson_m}
			\B_m\sum_{n\in \Lambda} \C^{m-n} \u^n  = \omega_0^2 \u^m,
		\end{equation}
	for every $m\in \Lambda$, where $\C^m = \I[\widehat{\C}^\alpha](m).$	
	\end{prop}
\begin{proof}
	We can combine the fact from \Cref{lem:constant} that the eigenmodes are approximately constant on the resonators with the fact from \Cref{thm:res_periodic} that the modes are given by eigenvectors of the generalized quasiperiodic capacitance matrix to obtain a discrete-operator formulation of the subwavelength resonance problem. This was proved in \cite[Proposition 4.2]{pt_topological} and says that 	
	\begin{equation} \label{eq:main}
		\widehat{\C}^\alpha \sum_{m\in \Lambda} \u^m e^{\iu \alpha\cdot m} = \omega^2_0\sum_{m\in \Lambda}(\B_m)^{-1} \u^m e^{\iu \alpha\cdot m},
	\end{equation}
	for all $\alpha \in Y^*$. We can choose to pose the problem either in the real (spatial) variable $m$ or in the dual (momentum) variable $\alpha$. 
	In this framework, equation \eqref{eq:main}, posed in the dual variable $\alpha$, reads
	\begin{align} \label{eq:anderson_alp}
		\widehat{\C}^\alpha \widehat{\u}^\alpha &= \omega_0^2 \left(\sum_{m\in \Lambda}(\B_m)^{-1} \I[\widehat{\u}^\alpha ](m)e^{\iu  \alpha \cdot m} \right), \qquad \widehat{\u}^\alpha =  \F[\u^m](\alpha),
	\end{align}
	We can re-frame this equation in the real variable $m$ by taking the inverse Floquet transform. Since $\u^m = \I[\widehat{\u}^\alpha](m)$, we can apply $\I$ to \eqref{eq:anderson_alp} to obtain \eqref{eq:anderson_m}.
\end{proof}

We can write \eqref{eq:anderson_m} in terms of discrete lattice operators. If $\uf \in \ell^2(\Lambda^N,\Cb^N)$, we define the operator $\Cf$ on $\ell^2(\Lambda^N,\Cb^N)$ by
$$\Cf \uf (m) = \sum_{n\in \Lambda} \C^{m-n} \uf(n).$$
Moreover, we define the operator $\Bf$ as the diagonal multiplication by $\B_m$. With this notation, equation \eqref{eq:anderson_m} reads
\begin{equation} \label{eq:anderson_f}
	\Bf\Cf \uf = \omega_0^2\uf.
\end{equation}
The equation \eqref{eq:anderson_f} characterizes the subwavelength resonant frequencies and their corresponding resonant modes. A localized mode corresponds to an eigenvalue of the operator $\Bf\Cf$. In the periodic case (when $\Bf = I$), the spectrum of the lattice operator $\Cf$ is continuous and does not contain eigenvalues, so there are no localized modes. The spectral bands are given by \Cref{thm:res_periodic} in this case. As we shall see, the operator $\Bf\Cf$ might have a pure-point spectrum in the non-periodic case.

For a linear lattice, $d_l=1$, these operators amount to doubly infinite matrices:
	$$\Cf = \left(\begin{smallmatrix} \sddots & \svdots & \svdots  & \svdots & \svdots & \sadots  \\
		\cdots & \C^0 & \C^1 & \C^2& \C^3  & \cdots \\
		\cdots & \C^{-1} & \C^0 & \C^{1}& \C^2  & \cdots \\
		\cdots & \C^{-2} & \C^{-1} & \C^0 & \C^1 & \cdots \\
		\cdots & \C^{-3} & \C^{-2} & \C^{-1} & \C^0 & \cdots \\
		\sadots & \svdots & \svdots  & \svdots & \svdots & \sddots \end{smallmatrix}\right), \quad \uf = \left(\begin{smallmatrix} \svdots \\ \u^{-1} \\ \u^{0}\\\u^{1}\\\u^{2}\\ \svdots\end{smallmatrix}\right), \quad \Bf = \left(\begin{smallmatrix} \sddots & \svdots & \svdots  & \svdots & \svdots & \sadots  \\
		\cdots & \B_{-1} & 0 & 0 & 0 & \cdots \\
		\cdots & 0 & \B_{0} & 0 & 0 & \cdots \\
		\cdots & 0 & 0 & \B_1 & 0 & \cdots \\
		\cdots & 0 & 0 & 0 & \B_2 & \cdots \\
		\sadots & \svdots & \svdots  & \svdots & \svdots & \sddots \end{smallmatrix}\right).$$
	Here, $\Cf$ is the (block) Laurent operator corresponding to the symbol $\widehat{\C}^\alpha$.

	\begin{rmk}
		This discrete lattice-operator formulation generalizes the approach used in \cite{pt_topological}.
	\end{rmk}	

	\subsection{Toeplitz matrix formulation for compact defects} \label{sec:toepltiz}
	
	In this section, we assume that the operators $B_m$ are identity for all but finitely many $m$. For simplicity, we phrase this analysis in the case $\Lambda = \Z$, where we assume that the resonators corresponding to $0 \leq m \leq M$ have differing parameters, \ie{},
	\begin{equation}\label{eq:b}
		b_i^m = \begin{cases} 1, \quad & m< 0 \text{ or } m > M, \\
		1 + \eta_i^m, & 0 \leq m \leq M,\end{cases}
	\end{equation}
	for some parameters $\eta_i^m \in (-1, \infty)$ which are independent of $\delta$. We let $H_m$ be the diagonal matrix with entries $\eta_i^m$. In this setting, we obtain a (block) Toeplitz matrix formulation. For $\omega \notin \sigma(\Cf)$, we define
	$$ \mathcal{T}(\omega) =  \left(\begin{smallmatrix} T^0 & T^1 & T^2 & \cdots & T^M   \\
		T^{-1} & T^0 & T^{1} & \cdots & T^{M-1}   \\
		T^{-2} & T^{-1} & T^0 & \cdots & T^{M-2}  \\
		\svdots & \svdots  & \svdots & \sddots & \svdots  \\
		T^{-M} & T^{-(M-1)} & T^{-(M-2)} & \cdots & T^0 \\ \end{smallmatrix}\right), \qquad T^m = -\frac{1}{|Y^*|}\int_{Y^*} e^{\iu \alpha m}\widehat{\C}^\alpha\left(\widehat{\C}^\alpha - \omega^2I\right)^{-1}\dx \alpha.$$
	
	\begin{prop} \label{prop:toeplitz}
		Assume that $\Lambda = \Z$ and that $b_i^m$ are given by \eqref{eq:b}. Then $\omega_0 \notin \sigma(\Cf)$ is an eigenvalue of \eqref{eq:anderson_f} if and only if 
		$$\det\bigl(I-\mathcal{H} \T(\omega_0) \bigr) = 0,$$
		where $\mathcal{H}$ is the block-diagonal matrix with entries $H_m$.
	\end{prop}
\begin{proof}
From \eqref{eq:anderson_alp} it follows that, for each $m$, we have
	$$\I\big[\B_m \widehat{\C}^\alpha \widehat{\u}^\alpha \big](m)  = \omega_0^2  \I\big[\widehat{\u}^\alpha\big](m).$$
	Since $H_m$ is only nonzero for $m=0,1,...,M$, we have
	\begin{align}
		\I\Big[ (\widehat{\C}^\alpha   - \omega_0^2I)  \widehat{\u}^\alpha + H_m\widehat{\C}^\alpha\widehat{\u}^\alpha \Big](m) = 0, \qquad &0\leq m \leq M, \label{eq:fourierlow}\\
		\I\Big[(\widehat{\C}^\alpha   - \omega_0^2I)  \widehat{\u}^\alpha\Big](m) = 0, \qquad &m< 0 \text{ or } m > M.
	\end{align}
	In other words, $(\widehat{\C}^\alpha   - \omega_0^2I)  \widehat{\u}^\alpha$ is given by a finite Fourier series in $\alpha$:
	$$(\widehat{\C}^\alpha   - \omega_0^2I)  \widehat{\u}^\alpha = \sum_{m=0}^M \mathbf{c}_m e^{\iu \alpha m},$$
	where $\mathbf{c}_m = 	\I[H_m\widehat{\C}^\alpha\widehat{\u}^\alpha](m)$. Together with \eqref{eq:fourierlow} we find that
	$$\mathbf{c}_n = -	H_n\sum_{m=0}^M\frac{1}{|Y^*|}\int_{Y^*}e^{\iu \alpha (m-n)} \widehat{\C}^\alpha(\widehat{\C}^\alpha   - \omega_0^2I)^{-1}  \mathbf{c}_m  \dx \alpha.$$
	If we denote $\mathbf{c} = \left(\begin{smallmatrix} \mathbf{c}_{0} \\ \svdots \\ \mathbf{c}_M\end{smallmatrix}\right)$, we then obtain the following non-linear eigenvalue problem in $\omega_0$,
	\begin{equation}\label{eq:toeplitz}
		\mathcal{H} \T(\omega_0) \mathbf{c} = \mathbf{c},
	\end{equation}
	where $\mathcal{H}$ is the block-diagonal matrix with entries $H_m$.
\end{proof}
	\begin{rmk}
	The localized modes can, equivalently, be described by the Toeplitz matrix with entries
		$$\widetilde{T}^m = \frac{1}{|Y^*|}\int_{Y^*} e^{\iu \alpha m}\left(\widehat{\C}^\alpha - \omega^2I\right)^{-1}\dx \alpha,$$
		leading to the non-linear eigenvalue problem
		$$\omega^2\mathcal{Y} \widetilde{\T}(\omega_0) \mathbf{c} = \mathbf{c},$$
		where $\mathcal{Y} = (I+\mathcal{H})^{-1} - I$.
	\end{rmk}

	\begin{rmk}
		The analysis in this section generalizes approaches found in \cite{edge_mode} and \cite{pt_defect, ln_defect}.
	\end{rmk}

	\section{Finite perturbations of a chain}\label{sec:defect}
	In this section, we take the specific example of a chain of equidistant resonators where a finite number of resonators are perturbed, as illustrated in Figures~\ref{fig:arraysketch}(a)~and~(b). Specifically, we have the lattice dimension $d_l = 1$, the space dimension $d=3$ and we take a unit cell with just a single resonator: $N=1$. As we shall see, a single defect may induce a one localized mode. When multiple defects hybridize, level repulsion will guarantee the existence of at least one localized mode. We emphasize that the phenomena highlighted in this section are not specific to the dimensionality and will be qualitatively similar in $d_l =2 $ or $d_l = 3$; see, for example, \Cref{fig:modesA}.
	
	\subsection{Single defect}
	
	We assume that $N = 1$ and that the resonator corresponding to $m=0$ has a different parameter, \ie{},
	\begin{equation}
	b_1^m = \begin{cases} 1, \ &m\neq 0, \\ 1 + \eta, \ &m=0, \end{cases}
	\end{equation}
	for some parameter $\eta > -1$. In the single-resonator case, $\widehat{\C}^\alpha$ is just a scalar $\widehat{\C}^\alpha = \lambda_1^\alpha$. It follows from \Cref{prop:toeplitz} that any eigenvalue $\omega$ satisfies $\eta T^0(\omega) = 1$; in other words that
	\begin{equation} \label{eq:M1}
		\frac{\eta}{|Y^*|}\int_{Y^*} \frac{\lambda_1^\alpha }{\omega^2-\lambda_1^\alpha} \dx \alpha = 1.
	\end{equation}
We let $\omega_* =  \max\limits _{\alpha\in Y^*}\sqrt{\lambda_1^\alpha}$. For $\omega > \omega_*$, it is clear that $T^0(\omega) >0$. Therefore, there are no solutions in the case $\eta\leq 0$. Moreover, 
	$$\lim_{\omega \to \omega_*}T^0(\omega) = \infty, \qquad \lim_{\omega \to \infty}T^0(\omega) = 0.$$
	Therefore, if $\eta>0$, there is a solution $\omega=\omega_0$ of \eqref{eq:M1}. In other words, the defect induces an eigenvalue $\omega_0^2$ in the pure-point spectrum of $\Bf\Cf$, corresponding to an exponentially localized eigenmode. 
	
	The expected localized eigenmode is shown in \Cref{fig:point}(a) for a finite system of 31 circular resonators. Here, the capacitance matrix was numerically computed using the multipole expansion method. This method relies on the fact that solutions to the Helmholtz equation on radially symmetric domains can be expressed concisely using Bessel functions for two-dimensional problems (and spherical harmonics for spherically symmetric domains in three dimensions). Hence, a discrete version of the operator $\S_D^\alpha$ can be obtained by considering a truncated (finite) subset of this basis and using appropriate addition theorems to handle the lattice sums. This was explained in detail in \cite[Appendix A]{ammari2020topologically}  for the three-dimensional case and in \cite[Appendix A]{davies2020hopf} for the two-dimensional case. Finally, it is straightforward to implement \Cref{defn:QCM} by calculating the inverse operator and integrating over the boundaries to obtain the capacitance coefficients $\widehat{C}^\alpha$. The software used to produce the examples presented in this work is available open source online\footnote{\url{https://doi.org/10.5281/zenodo.6577767}}.
	
\begin{figure}
	\centering
	\begin{subfigure}[b]{0.49\linewidth}
		\includegraphics[width=\linewidth]{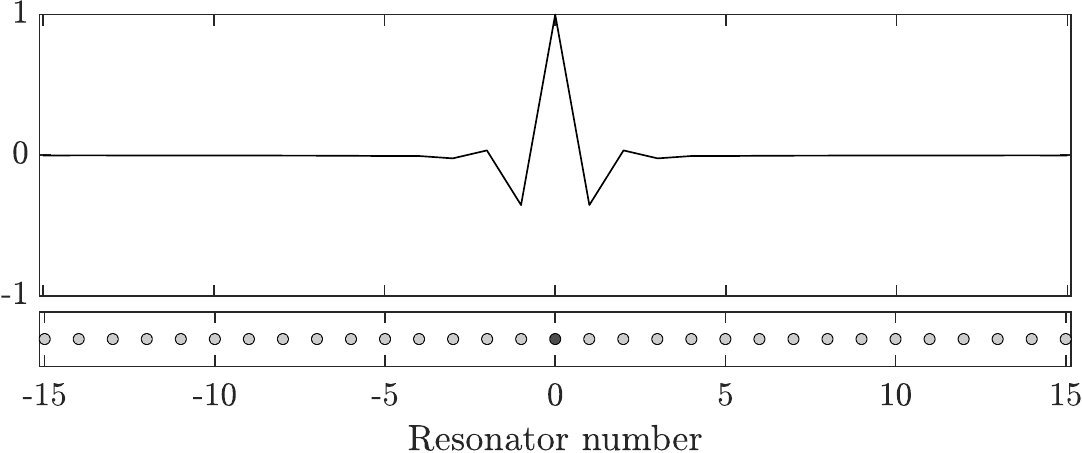}
		\caption{}\label{fig:pointa}
	\end{subfigure}
	\hspace{0.1cm}
	\begin{subfigure}[b]{0.49\linewidth}
		\includegraphics[width=\linewidth]{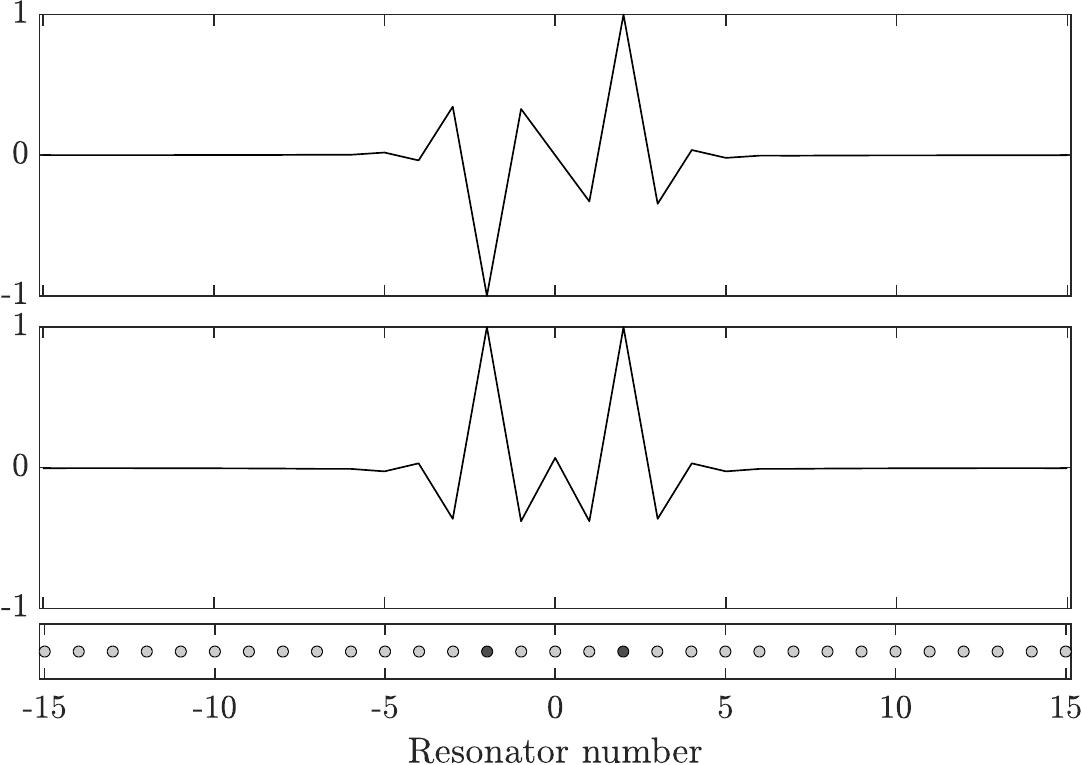}
		\caption{}
	\end{subfigure}
\caption{(a) Plot of the localized mode along the $x_1$-axis in the case of a single defect with $\eta = 0.15$. A sketch of the resonator structure is shown below: here we have a chain of $N = 31$ equally spaced circular resonators with the perturbation $1\mapsto 1+\eta$ made to the material parameters of centre resonator. We observe a single localized mode even in the case of small $\eta>0$, with larger $\eta$ leading to a stronger degree of localization (a faster decaying mode). However, in the case $-1<\eta <0$ there are no localized modes. (b) The two localized modes that exist when two defects are introduced with perturbation strengths $\eta_1=\eta_2=0.15$ and separation distance $M=4$. In this case, the higher frequency mode has a dipole (odd) symmetry whole the lower mode has a monopole (even) symmetry.} \label{fig:point}
\end{figure}

	\subsection{Double defect}
	Next, we continue with the case $N=1$ but now assume that there are two defects at the resonators corresponding to $m=0$ and $m=M$;
	\begin{equation}
	b_1^m = 1, \ m\notin \{0,M\},\qquad b_1^0 = 1 + \eta_1, \qquad b_1^M = 1+\eta_2,
	\end{equation}
	for some parameters $\eta_1,\eta_2\neq 0$. In this case, the localized modes are characterized by the Toeplitz matrix formulation
		$$\mathcal{H} \T(\omega) \mathbf{c} = \mathbf{c},$$
	where
	$$\mathcal{H} = \left(\begin{smallmatrix} \eta_1 & 0 & 0 & \cdots & 0   \\
		0 & 0 & 0 & \cdots & 0   \\
		0 & 0 & 0 & \cdots & 0  \\
		\svdots & \svdots  & \svdots & \sddots & \svdots  \\
		0 & 0 & 0 & \cdots & \eta_2 \\ \end{smallmatrix}\right), \quad \mathcal{T}(\omega) =  \left(\begin{smallmatrix} T^0 & T^1 & T^2 & \cdots & T^M   \\
		T^{-1} & T^0 & T^{1} & \cdots & T^{M-1}   \\
		T^{-2} & T^{-1} & T^0 & \cdots & T^{M-2}  \\
		\svdots & \svdots  & \svdots & \sddots & \svdots  \\
		T^{-M} & T^{-(M-1)} & T^{-(M-2)} & \cdots & T^0 \\ \end{smallmatrix}\right), \quad T^m = -\frac{1}{|Y^*|}\int_{Y^*} \frac{\lambda_1^\alpha e^{\iu \alpha m}}{\lambda_1^\alpha - \omega^2} \dx \alpha.$$
	In other words, we are solving the equation
	$$\det(I - \mathcal{H}\T(\omega)) = 0.$$
	We observe that
	$$\det(I - \mathcal{H}\T(\omega)) = \det\left(I-\begin{pmatrix} \eta_1T^0 & \eta_1T^M   \\
		\eta_2T^{-M} & \eta_2T^0 \\ \end{pmatrix} \right),$$
	so we seek the solutions $\omega$ to
	\begin{equation}\label{eq:2pt}
		\left(1-\eta_1T^0(\omega)\right)\left(1-\eta_2T^0(\omega)\right) - \eta_1\eta_2T^{-M}(\omega)T^M(\omega) = 0.
	\end{equation}

\begin{figure}
	\centering
	\begin{subfigure}{0.31\linewidth}
	\includegraphics[width=\linewidth]{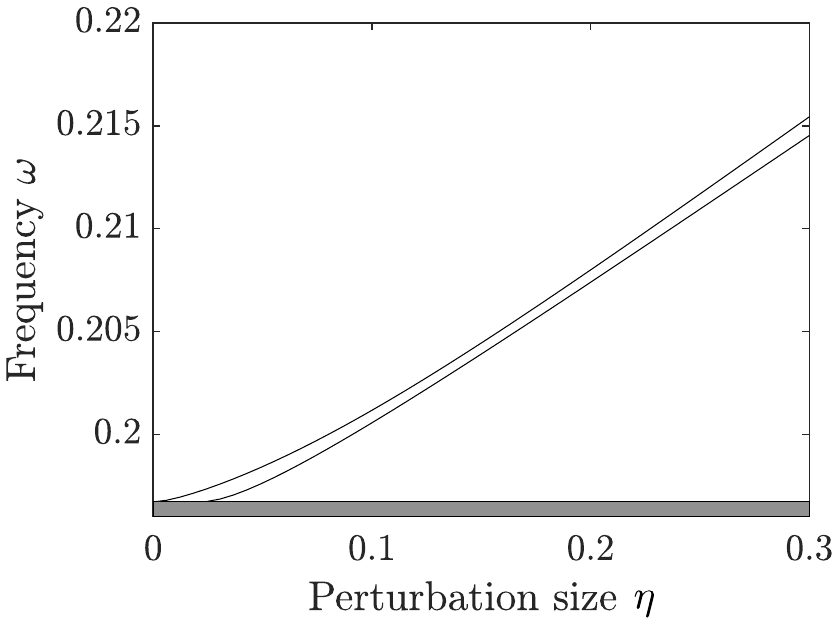}
	\caption{} \label{fig:funxM3}
	\end{subfigure}
	\hspace{0.1cm}
	\begin{subfigure}{0.31\linewidth}
	\includegraphics[width=\linewidth]{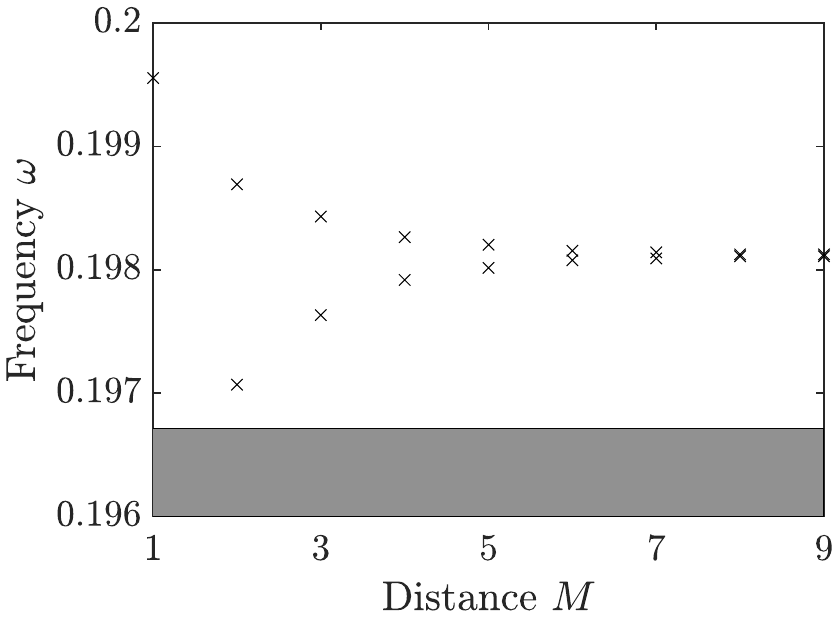}
	\caption{}
	\end{subfigure}
	\hspace{0.1cm}
	\begin{subfigure}{0.31\linewidth}
	\includegraphics[width=\linewidth]{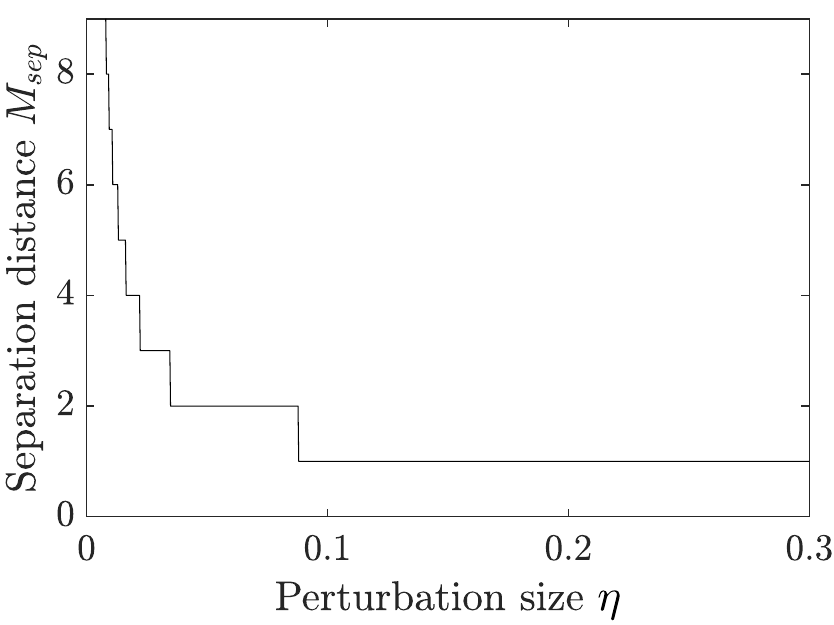}
	\caption{}
	\end{subfigure}
	\caption{Localized frequencies for two defects, for the case of identical defects $\eta_1=\eta_2=\eta$. (a) Solutions $\omega$ of \eqref{eq:2pt} as function of the perturbation size $\eta$ for fixed distance $M=3$. (b) Solutions $\omega$ of \eqref{eq:2pt} as function of the distance $M$ for fixed perturbation size $\eta=0.05$. If either $M$ or $\eta$ is small, then there is only one solution, as the level repulsion causes the second possible localized mode to fall within the spectral band (denoted by the shaded area). (c) The minimum separation distance $M_{sep}$ required to obtain two solutions to \eqref{eq:2pt}, plotted as a function of $\eta$ in the case $\eta_1=\eta_2 = \eta$.} \label{fig:twodefects}
\end{figure}

	\begin{figure}
	\centering
	\begin{subfigure}{0.45\linewidth}
	\includegraphics[width=\linewidth]{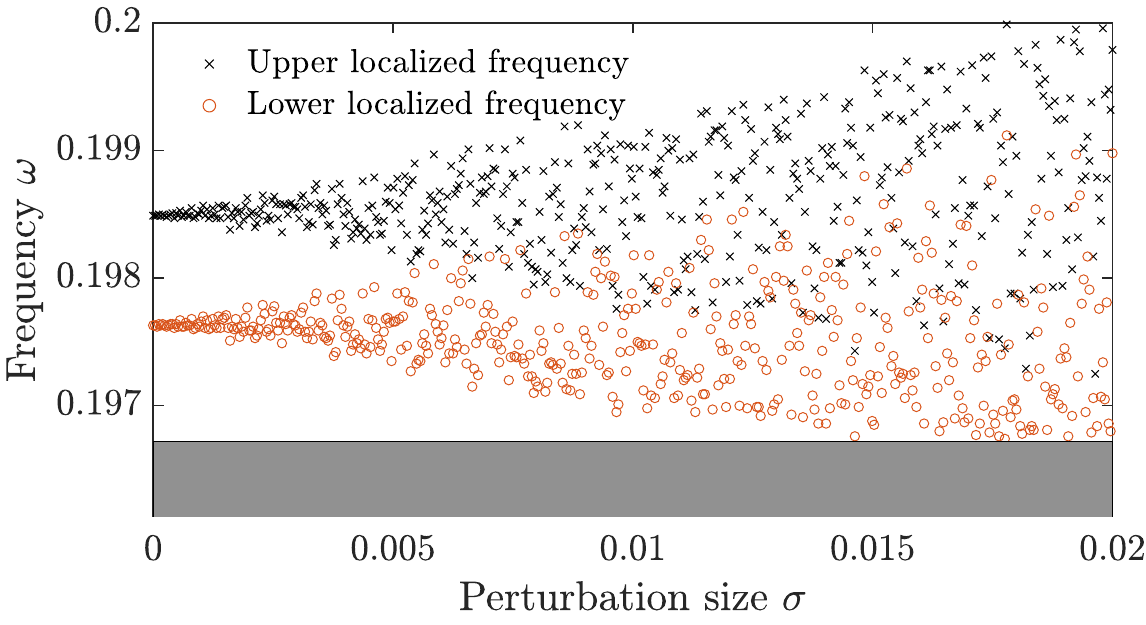}
	\caption{} \label{fig:tworandoma}
	\end{subfigure}
	\hspace{0.1cm}
	\begin{subfigure}{0.45\linewidth}
	\includegraphics[width=\linewidth]{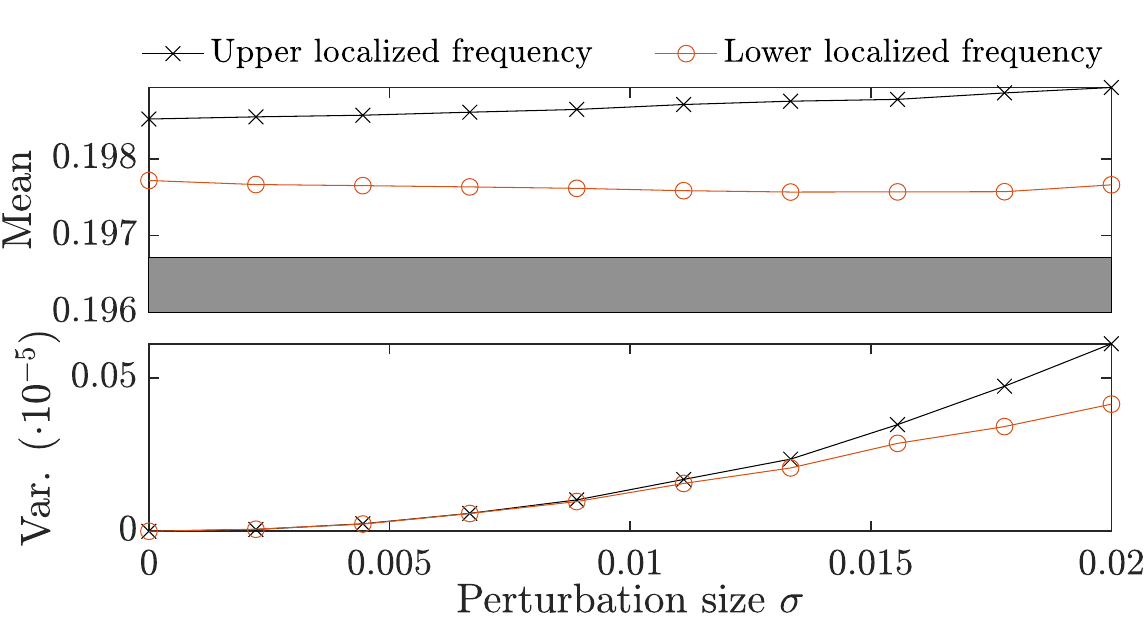}
	\caption{} \label{fig:two_meanvar}
	\end{subfigure}
	\caption{Localized frequencies for two random defects. (a) Solutions $\omega$ of \eqref{eq:2pt} with $\eta_1$ and $\eta_2$ drawn from a uniform $U[\mu-\sqrt{3}\sigma,\mu+\sqrt{3}\sigma]$ distribution, with varying $\sigma$, fixed mean $\mu=0.05$ and fixed separation distance $M=6$. (b) The mean (upper) and variance (lower) of the localized frequencies $\omega$ for different values of $\sigma$, fixed mean $\mu=0.05$ and fixed $M=3$. For each simulated value of $\sigma$, 1000 simulated values are used to compute the mean and variance, with values being left empty when no such frequency exists (\textit{e.g.} when there is only one solution to \eqref{eq:2pt}, as the hybridization causes the second solution to fall within the spectral band).} \label{fig:tworandom}
	\end{figure}
	
	\begin{figure}
	\centering
	\begin{subfigure}{0.45\linewidth}
	\includegraphics[width=\linewidth]{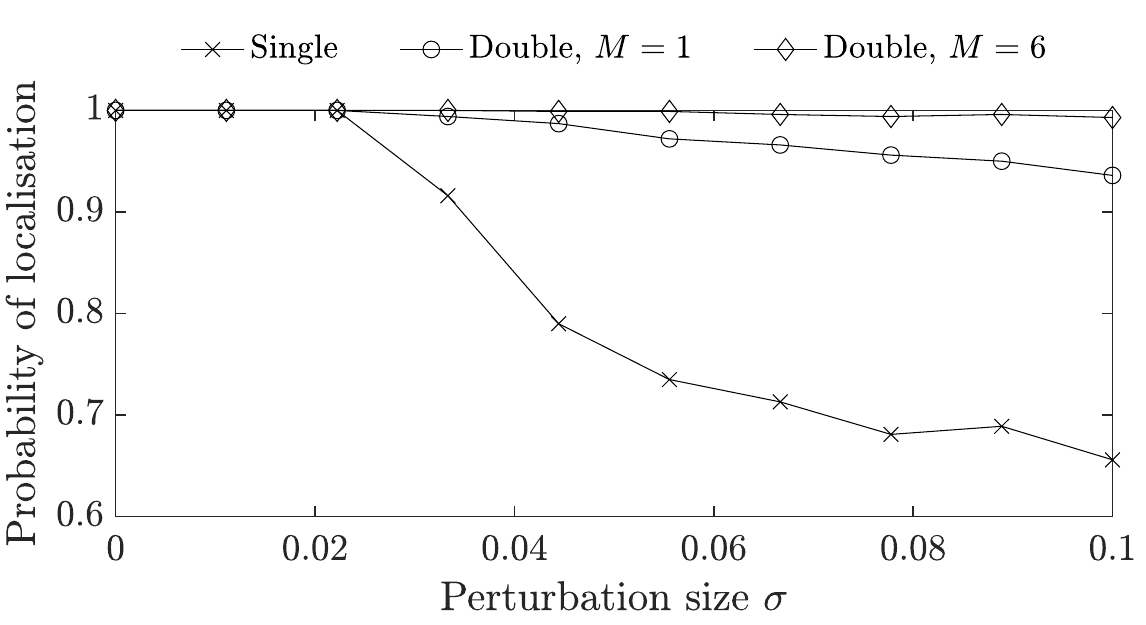}
	\caption{}
	\end{subfigure}
	\hspace{0.1cm}
	\begin{subfigure}{0.45\linewidth}
	\includegraphics[width=\linewidth]{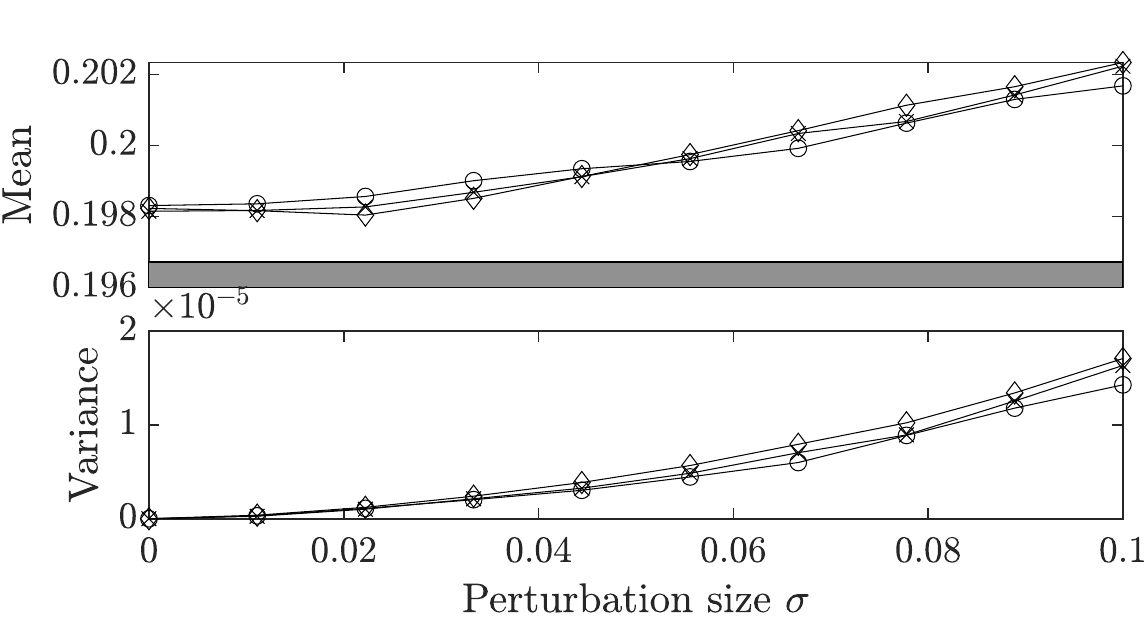}
	\caption{}
	\end{subfigure}
	\caption{Comparison of the robustness of systems with one or two defects. (a) The probability of at least one localized mode existing for the case of a single defect, a double defect with separation distance $M=1$ and a double defect with separation distance $M=6$. In each case, the perturbations are drawn at random from a uniform distribution  $U[\mu-\sqrt{3}\sigma,\mu+\sqrt{3}\sigma]$ and, for very large values of $\sigma$ we enforce $1+\eta_i\geq 10^{-4}>0$, so that the material parameters are all strictly positive. The means are chosen as $\mu=0.05$, $\mu=0.0325$ and $\mu=0.05$, respectively, so that the largest localized eigenfrequency is the same in each case when $\sigma=0$. (b) The mean and variance of the largest localized eigenfrequency for the three different models. Each value was calculated using 1000 independent random realizations.
	} \label{fig:robustness}
	\end{figure}

 	In \Cref{fig:twodefects} we show the resonant frequencies of defect modes in a system with two defects. These are obtained by solving \eqref{eq:2pt} for $\eta_1=\eta_2=\eta$ and varying either $\eta$ or the separation distance $M$. We see that for any $\eta>0$ there is at least one localized eigenmode. If either $M$ or $\eta$ is small, then there is only one solution, as the hybridization means the smaller solution is not greater than the maximum of the Bloch band so there is no localized eigenmode. For any fixed $\eta$ there is a critical minimum value $M_{\text{sep}}$ of $M$ in order for two localized eigenmodes to exist. In \Cref{fig:twodefects}(c) we show the relation between the critical values for the distance $M$ and the perturbation strength $\eta$.
	
	\subsubsection{Level repulsion in a two-defect structure}
	In \Cref{fig:tworandom} we model the resonant frequencies of a system with two random defects. In particular, we assume that the values of $\eta_1$ and $\eta_2$ are drawn independently from the uniform distribution $U[\mu-\sqrt{3}\sigma,\mu+\sqrt{3}\sigma]$. We see that for small $\sigma$ there are still two localized modes, however for larger values of $\sigma$ the eigenfrequencies can be be lost into the spectral band. Nevertheless, looking at \Cref{fig:tworandom}(b) we see that the introduction of random perturbations causes the average value of each mid-gap frequency to move further apart (and further from the edge of the band gap). This is the phenomenon of level repulsion. One consequence of this is that the eigenmodes, on average, become more strongly localized as the variance $\sigma^2$ increases (since the strength of the localization increases when the resonant frequency is further from the edge of the band gap). 
	
	Another consequence of level repulsion is that the existence of localized modes is much more robust in the case of two defects than in the case of a single defect. \Cref{fig:robustness}(a) shows the probability of at least one localized eigenmode existing for three different cases: a single defect, a double defect with small separation distance $M=1$ and a double defect with larger separation distance $M=6$. The defects are uniformly distributed with means chosen to normalize the distributions in the sense that the largest defect frequency when $\sigma=0$ is the same in all three cases ($\omega=0.198$). We see that the probability of localization drops off quickly for a single defect. When $\sigma=0.1$, the probability of localization for a single defect is 65.6\%, compared to 93.6\% for the double defect with small separation distance and 99.3\% when the separation distance is larger. The intuitive explanation for the difference between the two double defect cases is that the larger distance causes the two localized eigenfrequencies to be closer together in the deterministic case of $\sigma=0$, meaning that they interact with one another more strongly so the level repulsion has a greater effect. When at least one localized eigenmode exists, the distributions of the largest localized eigenvalue, shown in \Cref{fig:robustness}(b), are similar in all three cases.

	\subsubsection{Phase transition and eigenmode symmetry swapping}
	Analogously to \Cref{fig:twodefects}(a), \Cref{fig:trans}(a) shows the two localized frequencies, but here in the case $M=2$. When $M$ is even, the two curves will intersect (here around $\eta\approx 0.21$), corresponding to a doubly degenerate frequency. Crucially, this is a transition point whereby the symmetries of the corresponding eigenmodes swap. Below the transition point, the first mode has dipole (odd) symmetry while the second mode has monopole (even) symmetry; above the transition point the first mode has monopole (even) symmetry while the second mode has dipole (odd) symmetry.
	
	We now study the degree of localization around this transition point. Throughout, we define the degree of localization $l(u)$ of an eigenmode $u\in \ell^2(\Lambda^N,\Cb^N)$ as
	\begin{equation}
		l(u) = \frac{\|u\|_\infty}{\|u\|_2}.
	\end{equation}
	We remark that $0<l(u) \leq 1$: small $l(u)$ corresponds to a delocalized mode while $l(u) = 1$ corresponds to a mode which is perfectly localized to a single resonator. 
	
	In \Cref{fig:trans}(b) we plot the average degree of localization when the parameters are drawn independently from $\U[\eta-\sqrt{3}\sigma, \eta+\sqrt{3}\sigma]$ for $\sigma = 10^{-4}$. The plot shows a sharp peak at the transition point; a remarkably high degree of localization is achieved even when the variance is small. This can be understood as a decoupling of the (otherwise coupled) modes, which, at the transition point, become concentrated to the individual defects as opposed to both defects. 
	
	For higher values of $\sigma$, we expect to achieve a strong degree of localization for a larger range of $\eta$. As seen in \Cref{fig:trans}(c), there is a region in the $\eta$-$\sigma$-plane where we achieve optimal localization. When $\eta$ is large relative to $\sigma$, the modes will couple and the degree of localization is low. When $\sigma$ is large relative to $\eta$, we are in a regime with, on average, only a single localized mode and a low degree of localization.
	
	\begin{figure}
		\centering
		\begin{subfigure}{0.31\linewidth}
			\begin{tikzpicture}
			\draw  (0,0) node{\includegraphics[width=\linewidth]{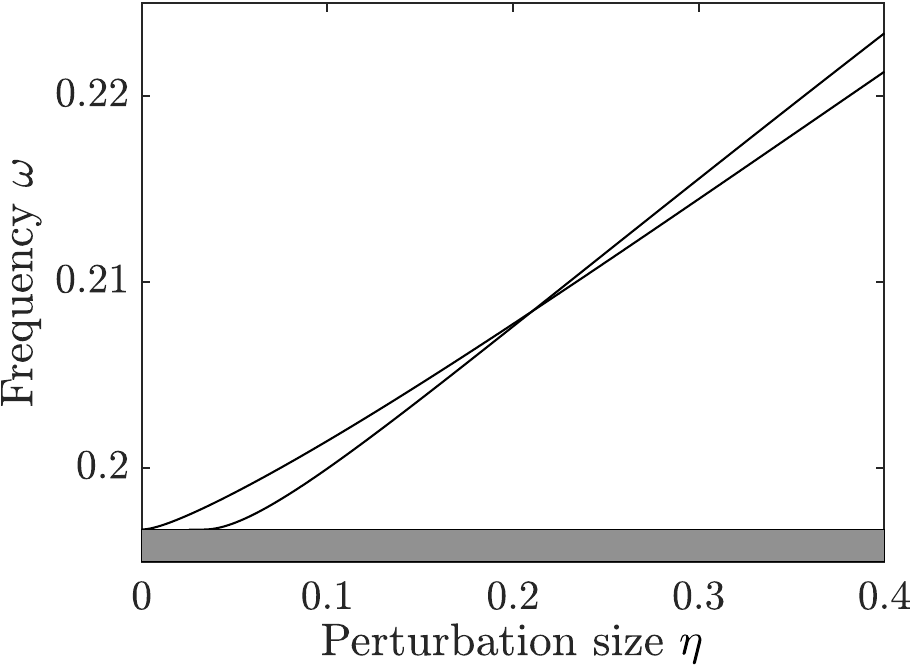}};

			\draw (0.4,-0.75)  node{\includegraphics[width=0.2\linewidth]{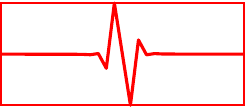}};
			\draw[->,red]  (-0.65,-0.7) to[out=0,in=180] (-0.1,-0.75);
			\draw[red,fill]  (-0.65,-0.7) circle(1pt);

			\draw (-1,0.1)  node{\includegraphics[width=0.2\linewidth]{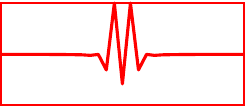}};
			\draw[->,red]  (-0.78,-0.62) to[out=100,in=-90] (-1,-0.12);
			\draw[red,fill]  (-0.78,-0.62) circle(1pt);

			\draw (1.7,0.2)  node{\includegraphics[width=0.2\linewidth]{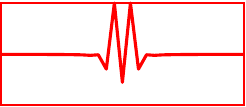}};
			\draw[->,red]  (1.6,0.92) to[out=-70,in=90] (1.7,0.42);
			\draw[red,fill]  (1.6,0.92) circle(1pt);

			\draw (0.7,1.3)  node{\includegraphics[width=0.2\linewidth]{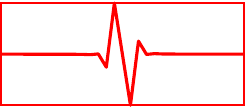}};
			\draw[->,red]  (1.5,1) to[out=90,in=0] (1.2,1.3);
			\draw[red,fill]  (1.5,1) circle(1pt);
			
			\end{tikzpicture}
			\caption{}\label{fig:M2a}
		\end{subfigure}
		\hspace{0.1cm}
		\begin{subfigure}{0.31\linewidth}
			\includegraphics[width=\linewidth]{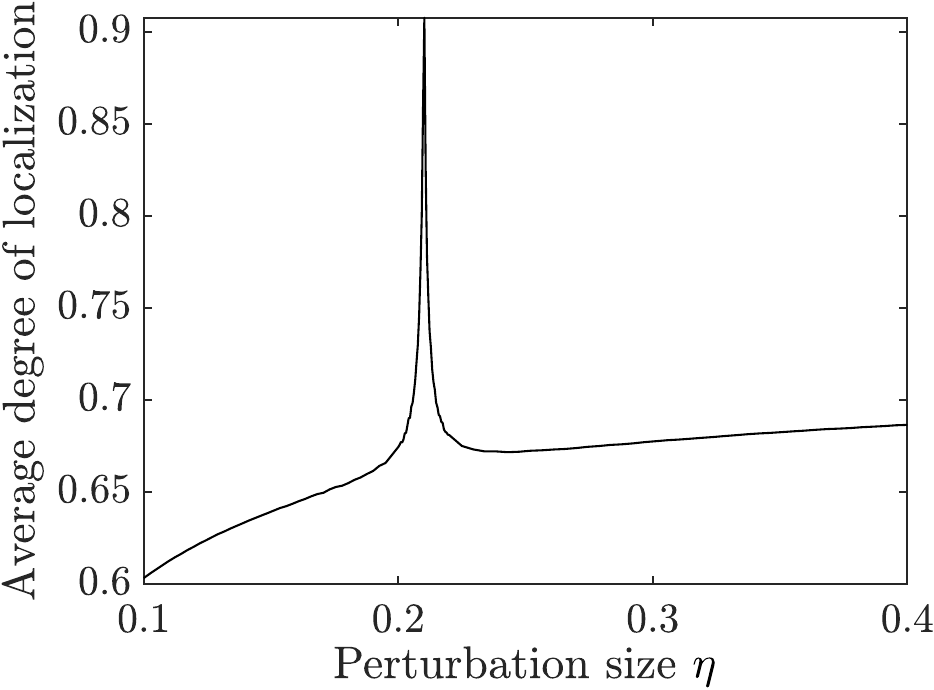}
			\caption{}\label{fig:M2b}
		\end{subfigure}
		\hspace{0.1cm}
		\begin{subfigure}{0.31\linewidth}
			\includegraphics[width=\linewidth]{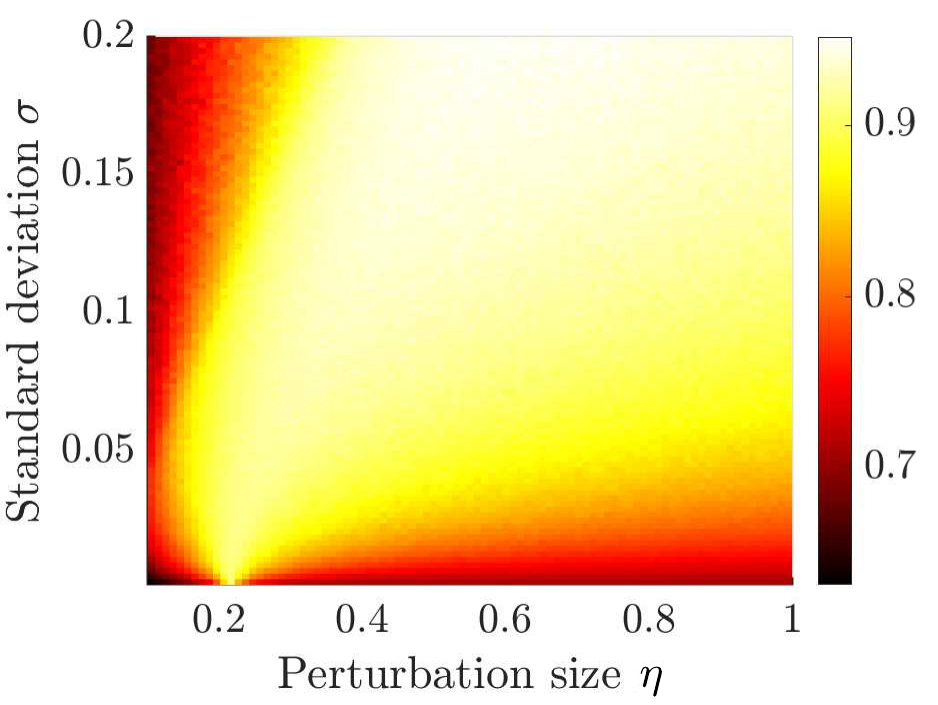}
			\caption{}\label{fig:M2c}
		\end{subfigure}
		\caption{Localized frequencies for two defects, for the case of separation distance $M=2$. (a) Solutions $\omega$ of \eqref{eq:2pt}, in the case of identical defects $\eta_1=\eta_2=\eta$, as function of the perturbation size $\eta$. The two curves intersect around $\eta\approx 0.21$. The symmetries of the eigenmodes (shown as inlays) are opposite on either side of the intersection. (b) Localization length $l(u)$ of the two modes, with parameters drawn independently from $\U[\eta-\sqrt{3}\sigma, \eta+\sqrt{3}\sigma]$ for $\sigma = 10^{-4}$. At the intersection, the monopole/dipole modes decouple, resulting in a sharp peak of $l(u)$. (c) Average degree of localization for different $\eta$ and $\sigma$. The peak seen in \textit{(b)} becomes wider for larger $\sigma$, giving rise to a region in the $\eta$-$\sigma$-plane with high degree of localization.} \label{fig:trans}
	\end{figure}
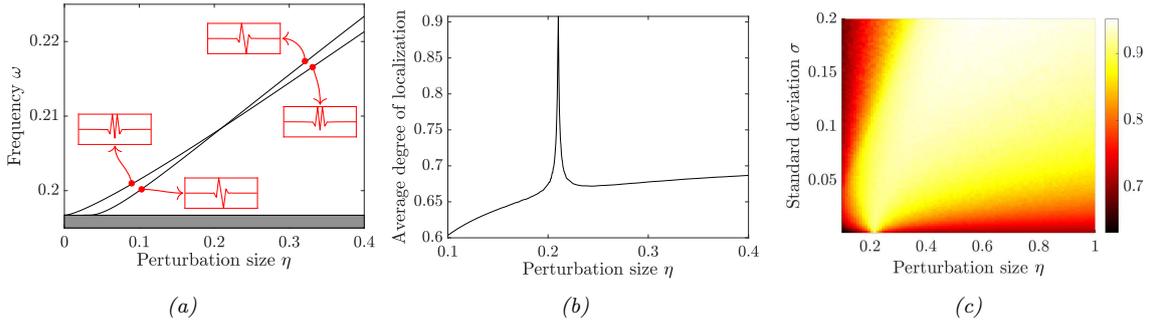

\subsection{Many defects}	

The motivation for the above analyses of systems with one or two local defects is to be able to infer properties of fully random systems, which will be studied in \Cref{sec:fullyrandom}. To further this analysis, we again assume that $N=1$ and consider a system with $N_d$ independent defects. That is, we assume that 

\begin{equation}
	b_1^m = 1, \ m\notin \{0,M,2M,\dots,N_dM\}\quad\text{and}\quad b_1^{kM} = 1 + \eta_{k}, \ k\in \{0,1,2,\dots,N_d\},
\end{equation}
for some parameters $\eta_1,\dots,\eta_{N_d}\neq 0$. In this case, we can use the Toeplitz matrix formulation from before to see that the localized modes are characterized by the equation
\begin{equation} \label{eq:generalN}
\det\left(I-\mathcal{A}_{\mathcal{HT}}(\omega) \right) = 0,
\end{equation}
where $\mathcal{A}_{\mathcal{HT}}(\omega) $ is the $N_d\times N_d$ matrix given by
\begin{equation*}
\mathcal{A}_{\mathcal{HT}}(\omega) _{ij}=\eta_i T^{M(j-i)}(\omega)\quad\text{with}\quad T^m = -\frac{1}{|Y^*|}\int_{Y^*} \frac{\lambda_1^\alpha e^{\iu \alpha m}}{\lambda_1^\alpha - \omega^2} \dx \alpha.
\end{equation*}
We can solve the system \eqref{eq:generalN} to find the localized modes for a system of $N_d$ defects. In \Cref{fig:manydefects} we show how the number of localized modes varies as the number of defects $N_d$ increases. The number of localized modes is normalized by dividing by the number of defects $N_d$, to give a measure of the density of localized modes (the number of defects is proportional to the length of the part of the resonator chain that contains the defects, which has length $N_dM$). Different values of $M$ and $\sigma$ are compared. In all cases, we observe that as $N_d$ becomes large the system appears to converge to a state whereby between $30\%$ and $50\%$ of the modes are localized.

	\begin{figure}
    	\centering
	\includegraphics[width=0.55\linewidth]{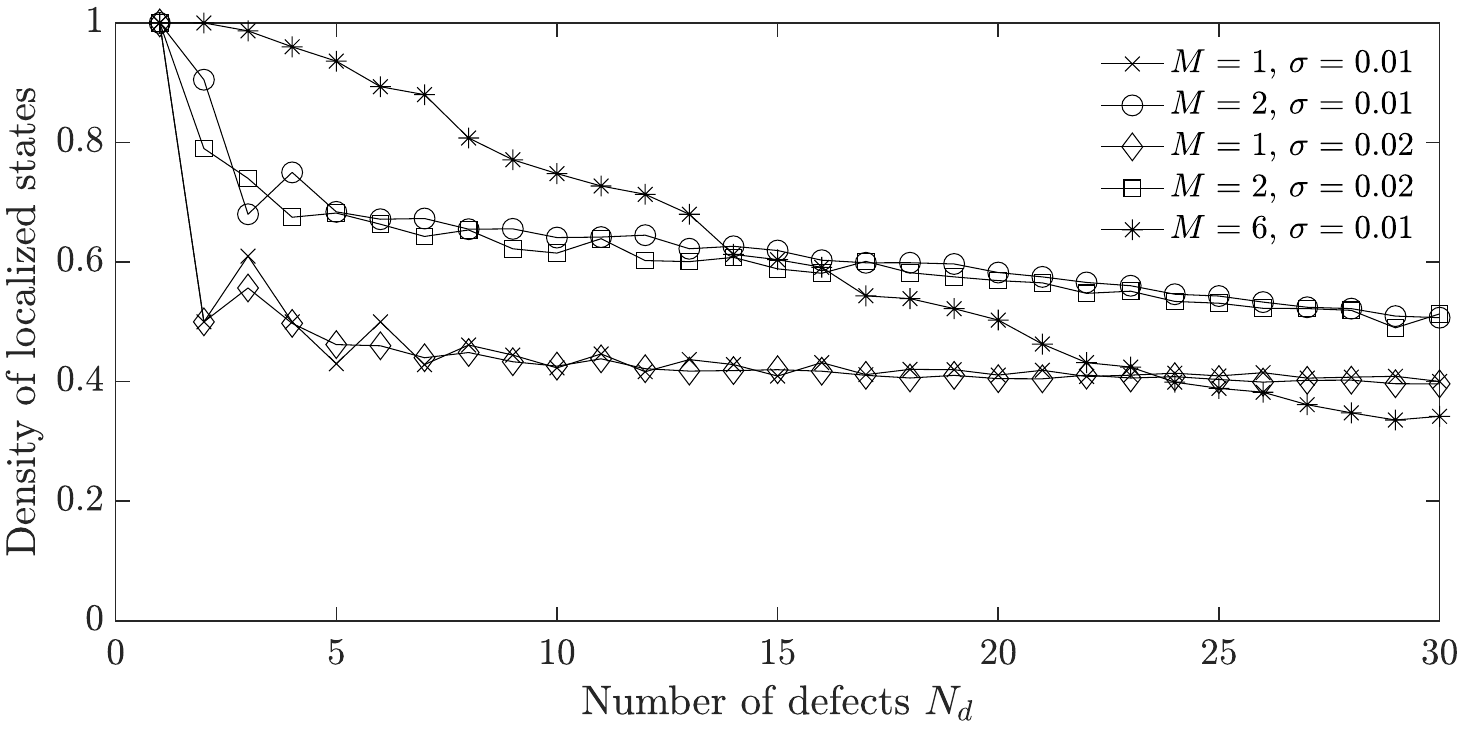}
    	\caption{Localized modes for increasingly many defects. We plot the density of licalized modes, given by the average number of localized modes (averaged over 100 independent random realizations) divided by the number of defects $N_d$. The defects are drawn from a uniform $U[0.05-\sqrt{3}\sigma,0.05+\sqrt{3}\sigma]$ distribution and have separation distance $M$. We show how the density of licalized modes varies as a function of the number of defects $N_d$ for several different values of $M$ and $\sigma$.} \label{fig:manydefects}
    	\end{figure}

	\section{Fully random parameter distribution} \label{sec:fullyrandom}
	Here we adopt the setting of \Cref{sec:defect} but additionally assume that all parameters $b_i^m$ are randomly distributed. We assume that
	$$b_i^m = 1 + \sigma \eta_i^m,\qquad \Bf = I + \sigma \Hf,$$
	where $\Hf$ is a diagonal operator. Here, we take $\eta_i^m$ to be independent and uniformly distributed on $[-\sqrt{3},\sqrt{3}]$. In order to guarantee that the material parameters remain positive, we require $\sigma < \frac{1}{\sqrt 3}$.
		
	Localized modes are characterized by the equation
	\begin{equation} \label{eq:C} 
		\left( \Cf + \sigma \Hf\Cf \right)\uf = \omega^2\uf. 
	\end{equation}
	In other words, we have characterized the problem in terms of the spectral problem for a random perturbation of a discrete operator. We emphasize that the unperturbed part of the operator, namely $\Cf$, is a dense matrix, whose off-diagonal terms describe long-range interactions with slow decay like $1/|i-j|$. Additionally, the perturbation $\sigma \Hf\Cf$ is also a dense matrix, containing random perturbations to both the diagonal- and off-diagonal terms. These two properties make the problem considerably more challenging than Anderson's original model, and precise statements on the spectrum of $\Cf + \sigma \Hf\Cf$ are generally beyond reach.

\subsection{Band structure and emergence of localized modes}

In \Cref{fig:dos}, we show the numerically computed density of states, along with the average degree of localization. Here, we take a chain of equally spaced resonators, which in the case $\sigma = 0$ has a single spectral band. Localized modes emerge as $\sigma$ becomes nonzero, whose frequencies lie around the edge of the spectral band. Increasing $\sigma$ leads to more localized modes and a stronger degree of localization. 
	
Next, we turn to the case $N=2$. Here, we take a chain of ``dimers'', \textit{i.e.} a chain of resonators with alternating, unequal, separation distance, as depicted in \Cref{fig:arraysketch}(d). In the case $\sigma = 0$, which was studied extensively in \cite{ammari2020topologically}, we have two spectral bands separated by a band gap. For small but nonzero $\sigma$, the band gap persists and localized modes emerge at either edge of the bands. Increasing $\sigma$ further, more modes become localized and the band-gap structure disappears.

	\begin{figure}
	\begin{center}
		\begin{subfigure}[b]{0.31\linewidth}
			\centering
			\includegraphics[width=\linewidth]{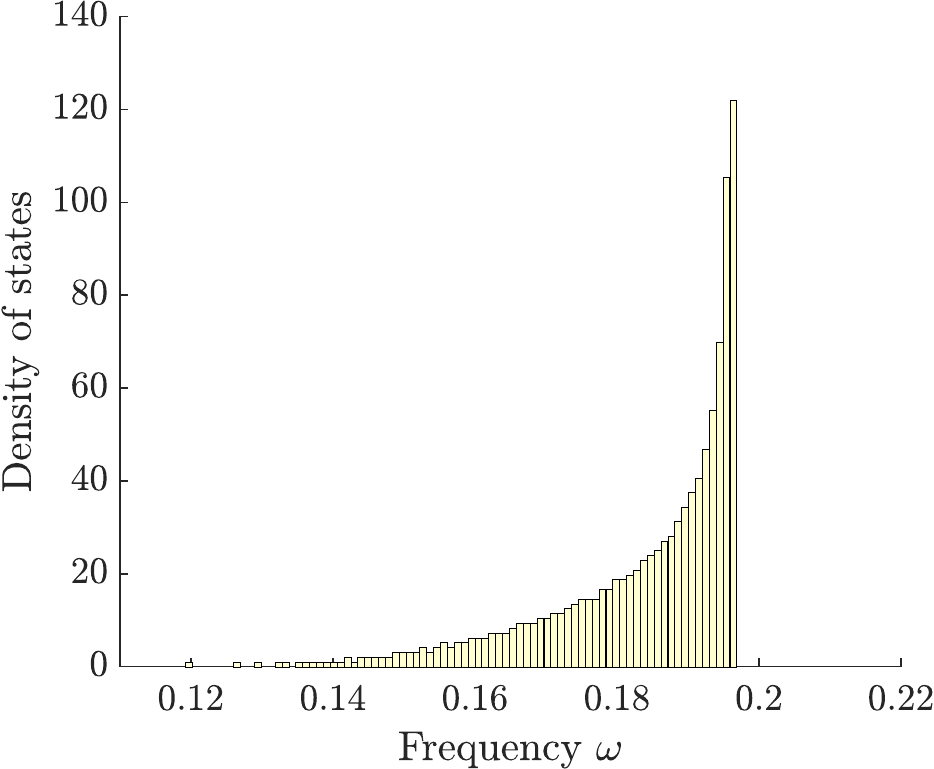}
			\caption{$\sigma = 0$}
		\end{subfigure}\hfill
		\begin{subfigure}[b]{0.31\linewidth}
			\centering
			\includegraphics[width=\linewidth]{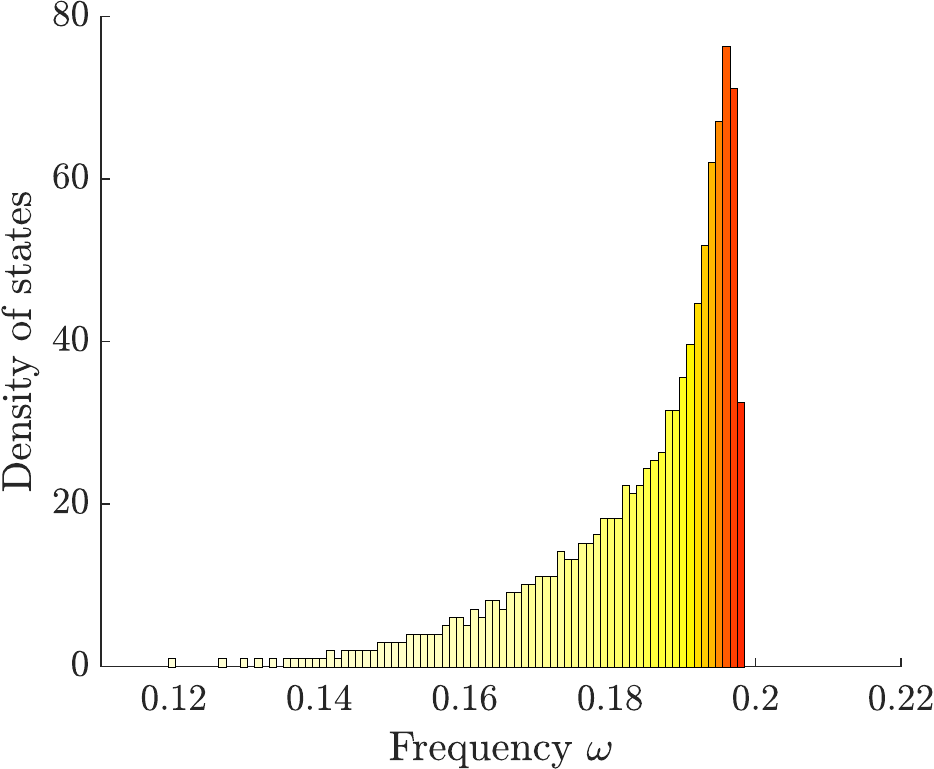}
			\caption{$\sigma = 0.02$.}
		\end{subfigure}\hfill
		\begin{subfigure}[b]{0.31\linewidth}
			\centering
			\includegraphics[width=\linewidth]{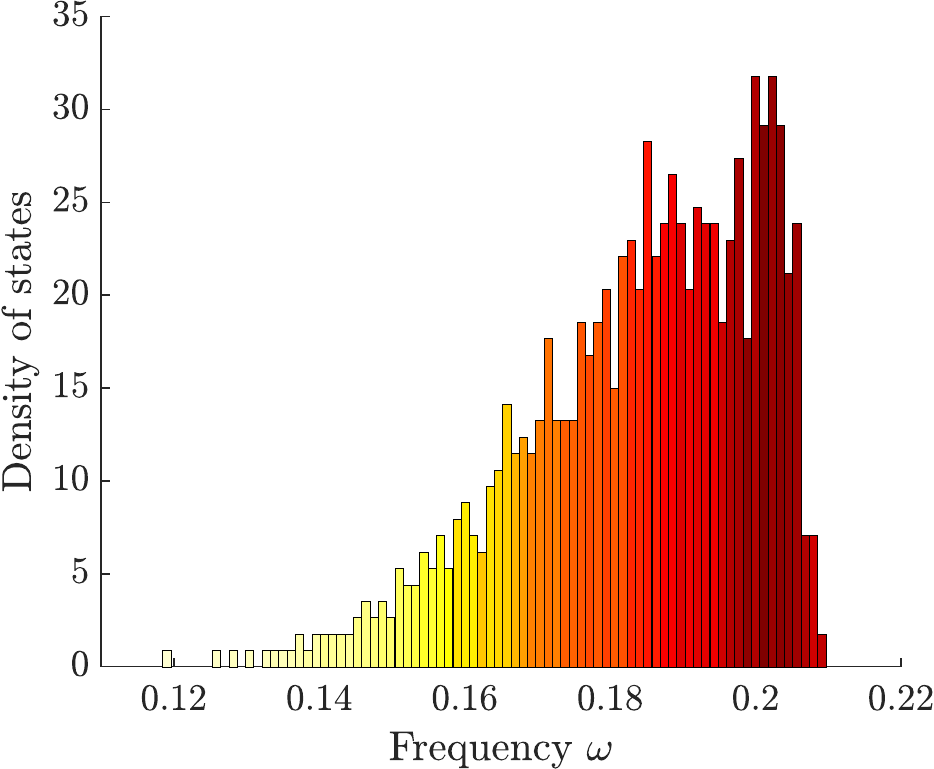}
			\caption{$\sigma = 0.1$.}
		\end{subfigure}
		\begin{subfigure}[b]{0.028\linewidth}
			\centering
			\includegraphics[width=\linewidth]{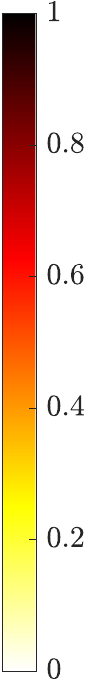}
			\vspace{0.6cm}
		\end{subfigure}
	\end{center}
	\caption{The density of states in a chain of equally spaced resonators with uniformly distributed random parameters. The colour of each bar in the histograms shows the average degree of localization $l(u) = \|u\|_\infty/\|u\|_2$ for the modes within that interval. A large value of $l(u)$ corresponds to a strongly localized mode $u$. In this case, there is a single spectral band, between $0$ and approximately $0.2$, when $\sigma = 0$ (which contains modes that are not localized). For small but nonzero $\sigma$, localized modes emerge at the edge of this band. The number of localized modes increase when $\sigma$ increases further. Here shown for a single realization of a finite chain of $N=1000$ resonators.} \label{fig:dos}
\end{figure}

	\begin{figure}
		\begin{center}
			\begin{subfigure}[b]{0.31\linewidth}
				\centering
				\includegraphics[width=\linewidth]{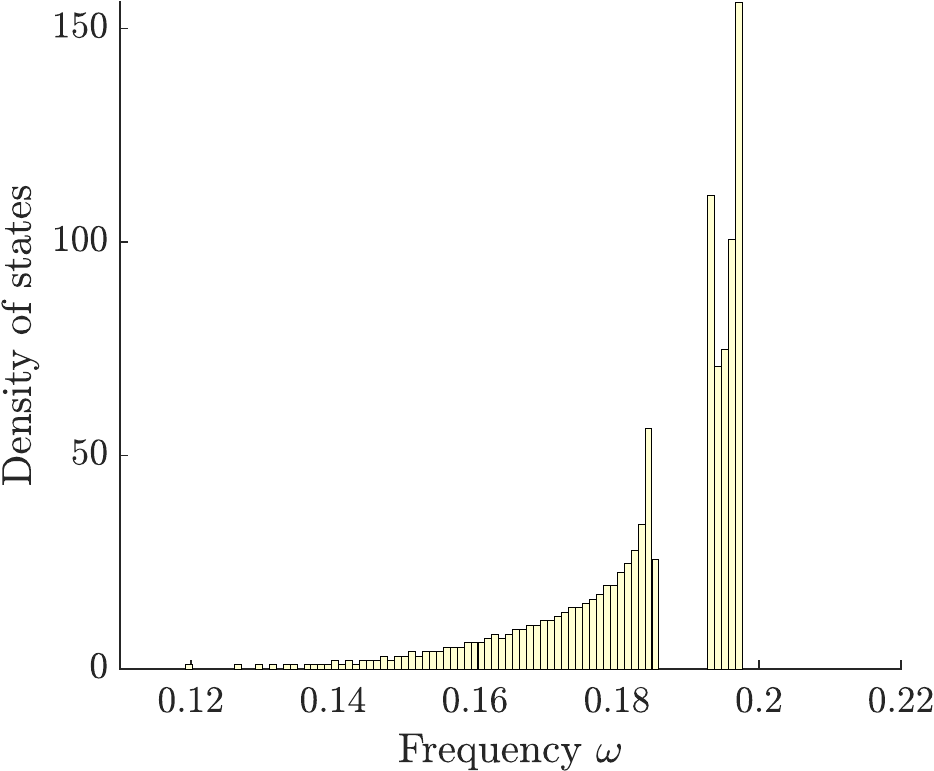}
				\caption{$\sigma = 0$}
			\end{subfigure}\hfill
			\begin{subfigure}[b]{0.31\linewidth}
				\centering
				\includegraphics[width=\linewidth]{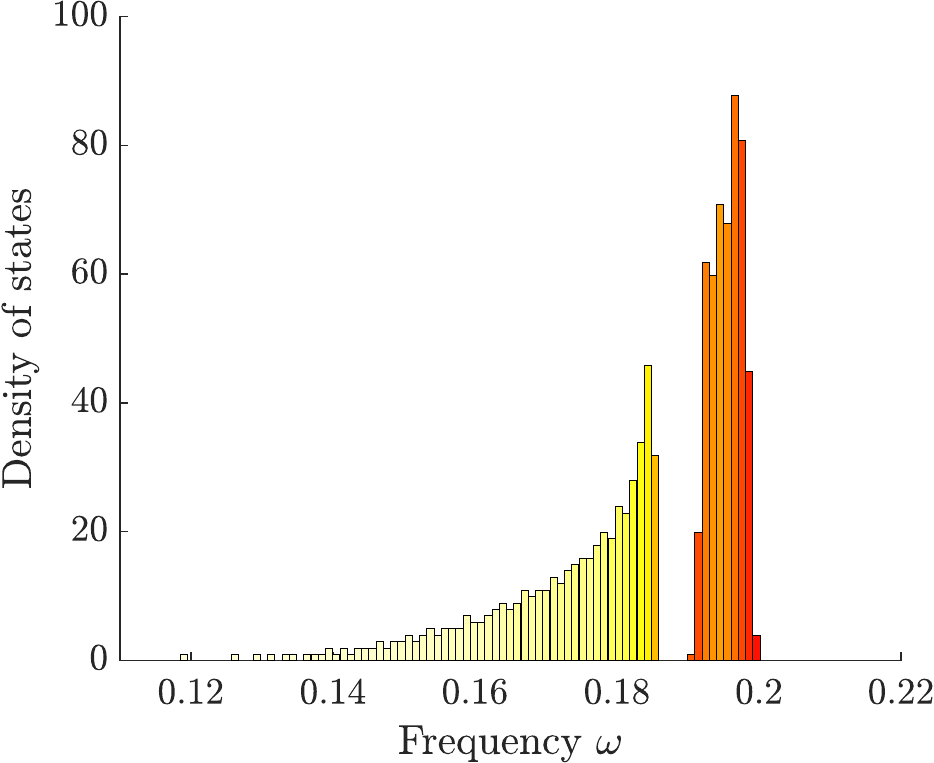}
				\caption{$\sigma =0.02$.}
			\end{subfigure}\hfill
		\begin{subfigure}[b]{0.31\linewidth}
		\centering
		\includegraphics[width=\linewidth]{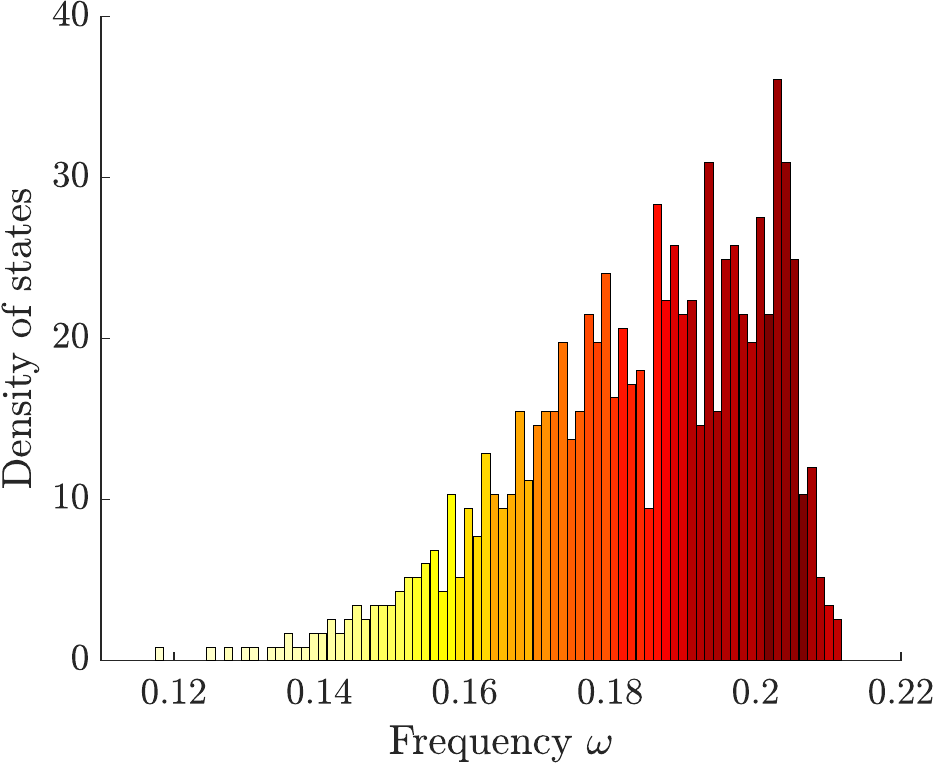}
		\caption{$\sigma = 0.1$.}
	\end{subfigure}
		\begin{subfigure}[b]{0.028\linewidth}
			\centering
			\includegraphics[width=\linewidth]{hot-eps-converted-to.pdf}
			\vspace{0.6cm}
		\end{subfigure}
		\end{center}
		\caption{The density of states in a chain of resonator dimers with uniformly distributed random parameters. The colour of each bar in the histograms shows the (normalized) average degree of localization $l(u)$. The case $\sigma=0$ exhibits two spectral bands, separated by a band gap (which both contain modes that are not localized). For small but nonzero $\sigma$, localized modes emerge around the edges of the bands: above the first band, and below and above the second band. When $\sigma$ increases, the band gap ultimately vanishes and the density of localized modes increases across the spectrum. Here, this is shown for a single realization of a finite chain of $N=1000$ resonators.} \label{fig:dos2}
	\end{figure}

\begin{figure}
	\centering
	\begin{subfigure}{0.31\linewidth}
	\includegraphics[width=\linewidth]{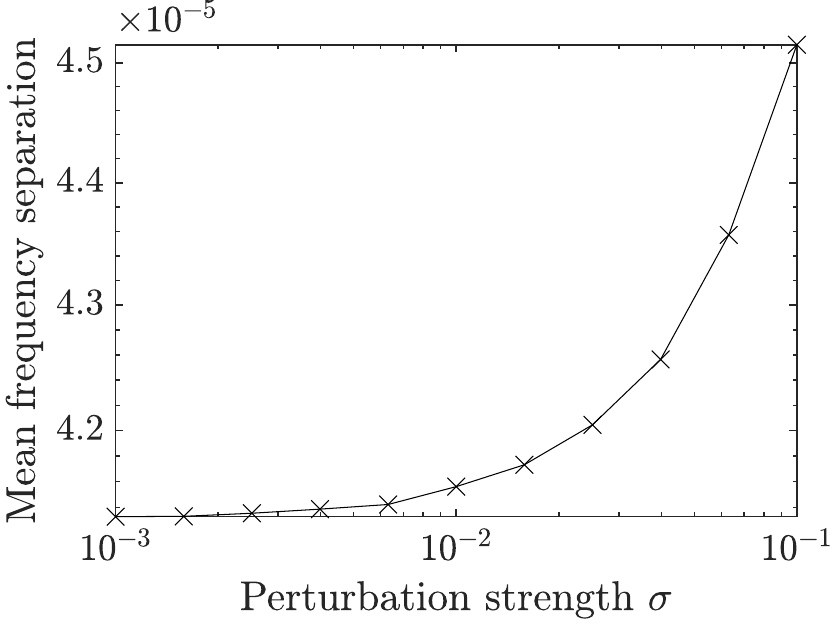}
	\caption{}
	\end{subfigure}
	\hspace{0.1cm}
	\begin{subfigure}{0.31\linewidth}
	\includegraphics[width=\linewidth]{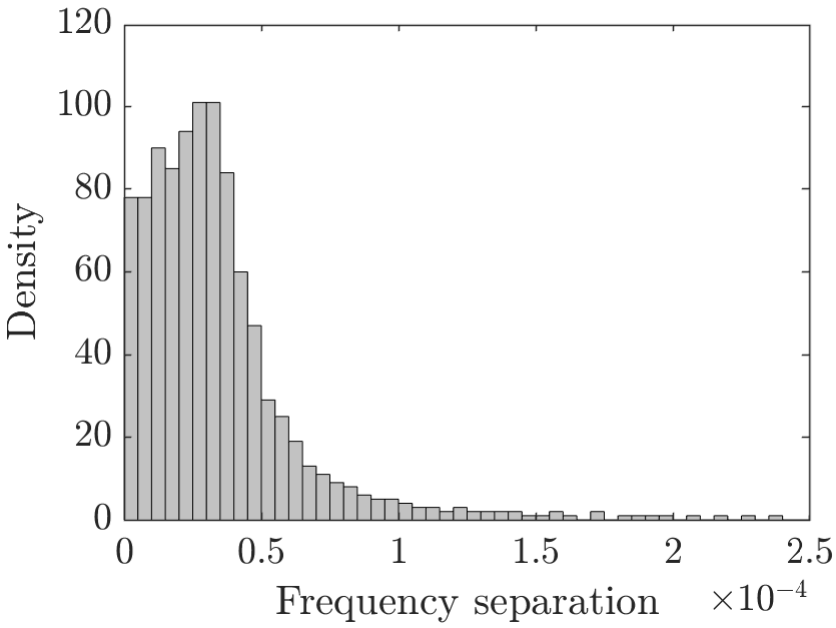}
	\caption{$\sigma=0.001$}
	\end{subfigure}
	\hspace{0.1cm}
	\begin{subfigure}{0.31\linewidth}
	\includegraphics[width=\linewidth]{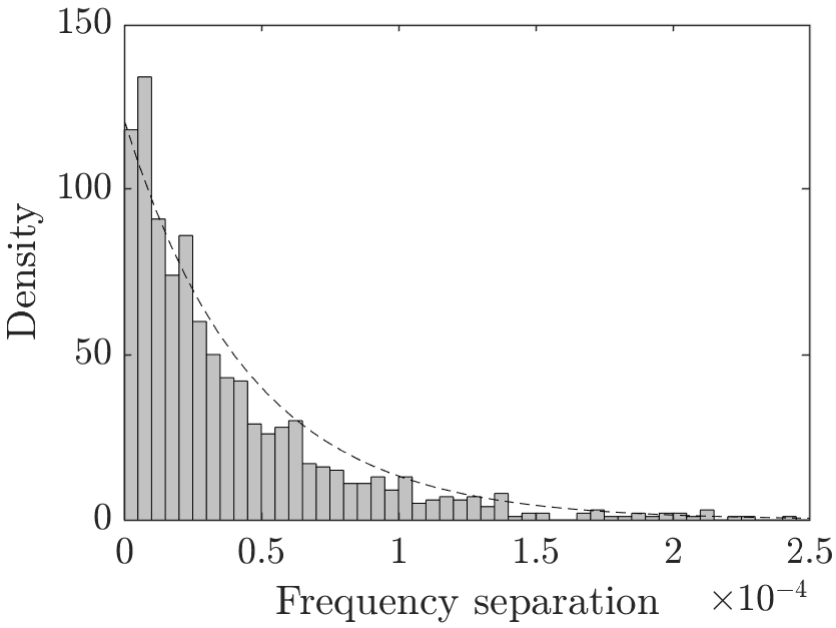}
	\caption{$\sigma=0.1$}
	\end{subfigure}
	\caption{Level repulsion in a finite chain of $N=1000$ resonators with random perturbations. (a) The average nearest-neighbour separation between the (real parts of the) resonant frequencies for different perturbation strengths $\sigma$, which is the standard deviation of the uniform random variables. (b) and (c) are histograms showing the distributions of the normalized frequency separations (frequency separation divided by mean frequency separation) for $\sigma=10^{-3}$ and $\sigma=10^{-1}$, respectively. A distribution of the form $\exp(-x)$ is shown with a dashed line on (c), for comparison.} \label{fig:separation}
\end{figure}

\subsection{Level repulsion in the fully random case}
Many of the arguments presented in this work involve the principle of level repulsion: that eigenvalues would tend to separate when increasingly large random errors are added. This was demonstrated in the double defect case in \Cref{fig:tworandom}(b). We can also study the level repulsion in a finite chain of equally spaced resonators. In particular, \Cref{fig:separation}(a) shows that the average separation increases as the perturbation strength $\sigma$ increases. We also show the distributions of the separations for two specific values of $\sigma$ in Figures~\ref{fig:separation}(b) and \ref{fig:separation}(c). In Figure~\ref{fig:separation}(b), the perturbations are so small that this distribution is (visually) the same as that of the unperturbed structure. In Figure~\ref{fig:separation}(c), the larger perturbations have affected the shape of the distribution. A distribution of the form $\exp(-x)$ is shown, for comparison, which is also seen in other randomly perturbed spectral problems (see \emph{e.g.} Figures~2~and~3 of \cite{brody1981random}). It is worth noting that much of the standard theory of eigenvalues of random matrices is not useful in this case. First, we are making a very specific choice of unperturbed structure, which is periodic and has spectral structure (which is the main factor determining the distribution in \Cref{fig:separation}(b)). Additionally, the random perturbations are encoded in the diagonal matrix $\Hf$, such that the entries of $\Cf + \sigma \Hf\Cf$ are not independent from one another (entries in the same row have non-zero covariance).

	\section{Conclusions}
	
	This work casts new light on the phenomenon of Anderson localization in systems with long-range interactions. Leveraging recent breakthroughs in high-contrast coupled-resonator theory, we have framed the key properties of this exotic phenomenon in terms of simple formulas and elementary physical principles.
	
	Using the generalized capacitance formulation, we have described Anderson localization in systems of high-contrast subwavelength resonators. Starting from first principles, we derived a discrete asymptotic approximation of the resonant states of the high-contrast Helmholtz problem. This takes the form of a discrete operator with long-range off-diagonal terms and a multiplicative perturbation factor. This can be reduced to a simple linear system in the case of compact defects, yielding exact formulas which reveal the fundamental mechanisms of subwavelength localization and explain the phenomena observed in random systems. Specifically, we characterized the localization in terms of level repulsion, which causes the eigenfrequencies to spread and corresponding modes to become localized. Additionally, we were able to identify a phase transition as an eigenmode swapping event and show that, at this phase transition, strong localization is achievable with an exceptionally small random perturbation. Finally, we studied fully random systems, where localized modes emerge around any edge of the band functions of the unperturbed structure and have properties that can be understood using our earlier results on compact defects.

	\section*{Acknowledgements}
	The work of HA was partly funded by the Swiss National Science Foundation grant number 200021--200307. The work of BD was partly funded by the EC-H2020 FETOpen project BOHEME under grant agreement 863179. The work of EOH was partially supported by the Simons Foundation Award No. 663281 granted to the Institute of Mathematics of the Polish Academy of Sciences for the years 2021-2023.
	
	\section*{Data availability}
	
	The code used in this study is available at \url{https://doi.org/10.5281/zenodo.6577767}. No datasets were generated or analysed during the current study.

	\bibliographystyle{abbrv}
	\bibliography{anderson}{}
	\end{document}